\documentclass{article}

\usepackage[nonatbib,final]{neurips_2020}

\usepackage[utf8]{inputenc} 
\usepackage[T1]{fontenc}    
\usepackage[colorlinks=true,linkcolor=black,citecolor=blue]{hyperref}       
\usepackage{url}            
\usepackage{booktabs}       
\usepackage{amsfonts}       
\usepackage{nicefrac}       
\usepackage{microtype}      

\usepackage{amsthm,amssymb, amsmath} 
\usepackage{color}
\usepackage{xspace}
\usepackage[draft]{todonotes}
\usepackage[utf8]{inputenc}
\usepackage{dsfont}
\usepackage{url}
\usepackage{braket}
\usepackage{caption}
\usepackage{tikz}
\usepackage{ifthen}
\usetikzlibrary{arrows}
\usepackage{geometry}
\usepackage{graphicx}
\usepackage{centernot}
\usepackage{mathtools}

\usepackage{subcaption}
\usepackage{tabularx}
\usepackage{multicol}
\usepackage{multirow}

\DeclareMathOperator*{\argmax}{arg\,max}

\newcommand{\wk}[1]{\textcolor{blue}{(wk: #1)}}

\usepackage{todonotes}

\theoremstyle{plain}
\newtheorem{theorem}{Theorem}[section]
\newtheorem{lemma}[theorem]{Lemma}

\newtheorem{proposition}[theorem]{Proposition}

\newtheorem{assumption}[theorem]{Assumption}
\theoremstyle{definition}

\theoremstyle{remark}

\newcommand{\proofstring}{Proof}
\renewenvironment{proof}[1][]{\noindent\textbf{\proofstring\ifthenelse{\equal{#1}{}}{:}{~(#1) :}}\xspace}{\hfill$\blacksquare$\medskip\par}

\newcommand{\R}{\mathds{R}}
\newcommand{\N}{\mathds{N}}

\newcommand{\E}{\mathds{E}}

\newcommand{\Prob}{\mathds{P}}
\newcommand{\Ind}{\mathds{1}}
\newcommand{\ind}{\Ind}

\renewcommand{\H}{\mathcal{H}} 
\newcommand{\tb}{\textbf} 

\title{A kernel test for quasi-independence}

\author{%
 Tamara Fern\'andez\\
  Gatsby Unit\\
 University College London\\
 \texttt{t.a.fernandez@ucl.ac.uk} \\
  \And 
  Wenkai Xu\\
  Gatsby Unit\\
 University College London\\
  \texttt{xwk4813@gmail.com} \\
  \And 
  Marc Ditzhaus  \\
Department of Statistics\\
TU Dortmund University\\
  \texttt{marc.ditzhaus@tu-dortmund.de} \\
\And 
  Arthur Gretton \\
  Gatsby Unit\\
  University College London\\
  \texttt{arthur.gretton@gmail.com } \\
}

\begin{document}
\maketitle

\begin{abstract}

We consider settings in which the data of interest correspond to pairs of ordered times, e.g, the birth times of the first and second child, the times at which a new user creates an account and makes the first purchase on a website, and the entry and survival times of patients in a clinical trial. In these settings, the two times are not independent (the second occurs after the first), yet it is still of interest to determine whether there exists significant dependence {\em beyond} their ordering in time. We refer to this notion as "quasi-(in)dependence".  For instance, in a clinical trial, to avoid biased selection, we might wish to verify that recruitment times are quasi-independent of survival times, where dependencies might arise due to seasonal effects. In this paper, we propose a nonparametric statistical test of quasi-independence. Our test considers a potentially infinite space of alternatives, making it suitable for complex data where the nature of the possible quasi-dependence is not known in advance.  Standard parametric approaches are recovered as special cases, such as the classical conditional Kendall's tau, and log-rank tests. The tests apply in the right-censored setting: an essential feature in clinical trials, where patients can withdraw from the study.  We provide an asymptotic analysis of our test-statistic, and demonstrate in experiments that our test obtains better power than existing approaches, while being more computationally efficient.
\end{abstract}

\section{Introduction}

Many practical scientific problems require the study of events which occur consecutively in time. We focus here on the setting where event-times, $X$ and $Y$, are only observed if they are in the ordered relationship $X\leq Y$. This type of data is commonly known as \emph{truncated data}, and, in particular, we say that $X$ is right-truncated by $Y$, or $Y$ is left-truncated by $X$. In clinical trails, for example, only patients still alive at the beginning  of the study can be recruited, hence the recruitment times $X$ and the survival times $Y$ are ordered. In the field of  insurance, a liability claim may be placed at a time $Y$ as a consequence of an incident at a time $X$. In e-commerce,  the time $Y$ of first purchase by a new user may only happen after the time $X$  when the user registers with the website. 

Our goal is to determine whether there exists an association between $X$ and $Y$ in the truncated data setting. 
Given that $X\leq Y$, the times $X$ and $Y$ clearly will not be independent (with the exception of trivial cases in which, for instance, $X$ and $Y$ have disjoint support). 
Thus, while it is not meaningful to test for statistical independence in the truncated setting, we can nevertheless still test for whether $X$ and $Y$ are uncoupled apart from the fact that $X\leq Y$, using the notion of {\em quasi-independence}. We will make this notion formal in Section \ref{Sec:Quasi}.  

Testing for an association between ordered $X$ and $Y$ may be important in making business/medical decisions. In the setting of clinical trials, it is important to ensure that survival times are as ``independent" as possible from recruitment times, in order to avoid bias in the recruitment process. 
In e-commerce, it may be of interest to test whether the purchase time for an item, such as a swimsuit, depends on the registration time, to determine seasonal effects on consumer behaviour and refine advertising strategies.
In statistical modelling, a common working assumption is that $X$ and $Y$ are independent, but can only be observed when $X\leq Y$ holds: see e.g, \cite{hyde1977testing, tsai1988estimation, Woodroofe1985}, and \cite[Chapter 9]{klein2006survival}.
The independence assumption can be weakened to quasi-independence, which is testable, and under which typical methods are still valid \cite{klein2006survival,LagakosBarrajDeGruttola1988,TsaiETAL1987,Turnbull1976,Wang1991,Woodroofe1985}. 

Our tests apply in the setting where $Y$ is right-censored. This is a very common scenario in real-world applications, particularly in clinical trials, where patients may withdraw from the study before their event of interest is observed. 
In the e-commerce example, there may be registered users that have not yet made a purchase when the study ends. Formally, the data corresponds to the triple $(X,T,\Delta)$, where $T=\min\{C,Y\}$ is the minimum between the survival time $Y$ of a given patient, and the time $C$ at which said patient leaves the study (or the study ends),  and $\Delta=\ind_{\{T=Y\}}$. Given the truncated data setting, we have further that $X\leq \min\{Y,C\}$. We emphasise that quasi-independence and right-censoring are very different data properties. Quasi-independence is a deterministic hard constraint $(X\leq Y)$, while right-censoring is a stochastic property of the data (incomplete observations). 

Quasi-independence has been widely studied in the statistics community, including for right-censored data: we provide a brief review below (more detailed descriptions of relevant concepts and methods will be provided in subsequent sections). 
 
In this work, we propose a non-parametric statistical test for quasi-independence, which applies under right censoring. Our test statistic is a nonparametric generalisation of the log-rank test proposed by \cite{EmuraWang2010}, where the departure from the null is characterised by functions in a reproducing kernel Hilbert space (RKHS).  Consequently, we are able to straightforwardly detect a very rich family of alternatives, including non-monotone alternatives. 
Our test generalises statistical tests of independence based on the Hilbert-Schmidt Independence Criterion \cite{gretton2008kernel}; which were adapted to the right-censoring setting in \cite{fernandez2019kernel,rindt2019nonparametric}. Due to the additional correlations present in the test statistic under quasi-independence, however, we will require new approaches in our analysis of the consistency and asymptotic behaviour of our test statistic, compared with these earlier works.  

\begin{figure}[t!]
  \centering
\includegraphics[width=1\textwidth]{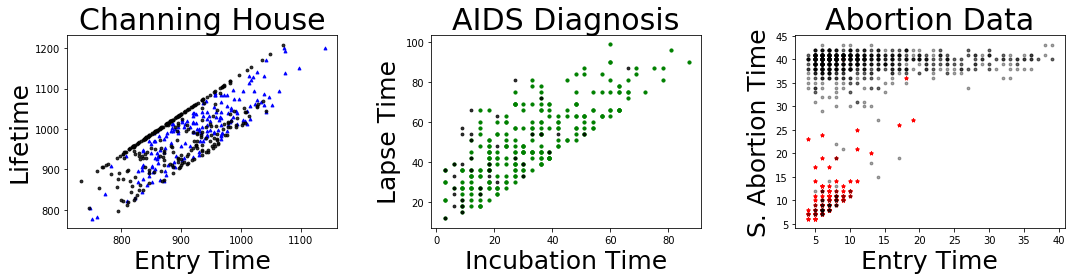}
\caption{Channing House dataset: the x-axis shows the entry time to the retirement center; and the y-axis shows the right-censored lifetimes . Events are censored by withdrawal from the center or study finishes at July 1, 1975. AIDS dataset: the x-axis shows the incubation time X; and the y-axis shows the censored lapse time Y, measured from infection to recruitment time. Events are censored by death or left the study.
Infected patients were recruited in the study only if they developed AIDS within the study period, therefore, in this dataset, the incubation time X does not exceed the lapse time Y. 
Abortion dataset: the x-axis shows the time to enter the study; and the y-axis shows the right-censored time for spontaneous abortion. Events are censored due to life birth and induced abortions. 
All censored times are marked in dark.}\label{Fig:1}
\vspace{-0.6cm}
\end{figure}
In Section \ref{Sec:Quasi}, we introduce the notion of quasi-independence. We next propose an RKHS statistic to detect this quasi-independence, and
its finite sample estimate from data.
We contrast the statistic for quasi-independence with the analogous RKHS statistic for independence, noting the additional sample dependencies on account of the left-truncation. Next, in Section \ref{sec:right-censoring}, we
generalise the quasi-independence statistics to account for the presence of 
right-censored observations. In Section \ref{Sec:asymptotic}, we provide our main theoretical results: an asymptotic analysis for our test statistic, and a guarantee of consistency under the alternative. In order to determine the test threshold in practice, we  introduce a Wild Bootstrap procedure to approximate the test threshold. In Section \ref{sec:Experiments}  we give a detailed empirical evaluation of our method.
We begin with challenging synthetic datasets exhibiting periodic quasi-dependence, as would be expected for example from seasonal or daily variations, where our approach strongly outperforms the alternatives. Additionally, we show our test is consistently the best test in data-scenarios in which the censoring percentage is relatively high, see Figure \ref{fig:censor_level}. 
Next, we apply our test statistic to three real-data scenarios, shown in Figure \ref{Fig:1}: a survival analysis study  for residents in the Channing House retirement community in Palo Alto, California \cite{hyde1977testing};  a study of transfusion-related AIDS \cite {LagakosBarrajDeGruttola1988}; and a spontaneous abortion study \cite{meister2008statistical}. For this last dataset, our general-purpose test is able to detect a mode of quasi-dependence discovered by a model that exploits domain-specific knowledge, but not found by alternative  general-purpose testing approaches. This was a particular challenge due to the large percentage of censored observations in the abortion dataset; see Figure \ref{fig:censor_level}. Proofs of all results are given in the Appendix.

\section{Quasi-independence}\label{Sec:Quasi}

Our goal is to infer the null hypothesis of quasi-independence between $X$ and $Y$. Formally, this null hypothesis is characterised as
\begin{align}
H_0: \pi(x,y)=\widetilde{F}_{X}(x)\widetilde{S}_{Y}(y),\quad \text{for all }x\leq y\label{eqn:q1},
\end{align}
where $\pi(x,y)=\Prob(X\leq x,Y\geq y)$, and  $\tilde{F}_X(x)$ and $\tilde{S}_Y(y)$ are functions that only depend on $x$ and $y$, respectively. In case of independent $X$ and $Y$, $\widetilde{F}_X(x)$ and $\widetilde{S}_Y(y)$ coincide with $F_{X}(x)=\Prob(X\leq x)$ and $S_Y(y)=\Prob(Y\geq y)$, but in general they may differ. For simplicity, $X$ and $Y$ are assumed continuously distributed on $\R_+$, and $f_{XY}$, $f_X$ and $f_Y$ denote the joint density and the corresponding marginals, and $f_{Y|X=x}$ denotes the conditional density of $Y$ given $X=x$. 
 
To simplify the notation, we suppose throughout that $X\leq Y$ always holds, and thus write $\pi(x,y)=\Prob(X\leq x,Y\geq y)$ instead of $\pi(x,y)=\Prob(X\leq x,Y\geq y|X\leq Y)$, as $\Prob(X\leq Y)=1$. We remark, however, that the ordering $X \leq Y$ can be ensured by considering a conditional probability space given $X \leq Y,$ and restricting calculations of probabilities, expectation etc. to this space; see \cite{chiou2018permutation,EmuraWang2010,Tsai1990}.

The notion of quasi-independence must not be confused with the notion of independent increments, i.e., $X\perp (Y-X)$. For instance, generate $X$ and $Y$ such that $X\leq Y$ by sampling i.i.d. uniform random variables, say $(U_1,U_2)$, in the interval $(0,1)$, and make $X=U_1$ and $Y=U_2$ for the first pair $(U_1,U_2)$ such that $U_1\leq U_2$. It can be verified that this construction leads to quasi-independent random variables $(X,Y)$, but $X$ and $Y-X$ are not independent as the distribution of $Y-X$ is constrained by how large the original value of $X$ was. The larger $X$ is, the smaller is the value of $Y-X$.

In \cite{EmuraWang2010}, the authors propose to measure quasi-independence by using a log-rank-type test-statistic which estimates $\int_{x \leq y}\omega(x,y)\rho(x,y)dxdy$, where 
\begin{align}\label{eqn:rho}
\rho(x,y)&=-\pi(x,y)\frac{\partial^2\pi(x,y)}{\partial x\partial y }+\frac{\partial\pi(x,y)}{\partial x}\frac{\partial\pi(x,y)}{\partial y}, \,x \leq y.
\end{align}
The function $\rho$ is originally inspired by the odds ratio proposed by \cite{chaieb2006estimating} (notwithstanding that $\rho$ is here a difference, rather than a ratio). Under the assumption of quasi-independence, $\rho=0$, and thus $\int_{x \leq y}\omega(x,y)\rho(x,y)dxdy=0$.  Nevertheless, it may be that  $\int_{x \leq y}\omega(x,y)\rho(x,y)dxdy=0$ even if the quasi-independence assumption is not satisfied, since the quantity depends on the function $\omega$: for instance, it is trivially zero when $\omega=0$. To avoid choosing a specific weight function $\omega$, we optimise over a class of weight functions, taking an RKHS approach,
\begin{align}\label{eq:kqic}
\Psi&=\sup_{\omega\in B_1(\mathcal{H})} \int_{x\leq y}\omega(x,y)\rho(x,y)dxdy,
\end{align}
where $B_1(\mathcal{H})$ is the unit ball of a reproducing kernel Hilbert space $\mathcal{H}$ with bounded measurable kernel given by $\mathfrak{K}:\R_+^2\times\R_+^2 \to \R$. We refer to the measure $\Psi^2$ as \emph{Kernel Quasi-Independent Criterion (KQIC)}. It can  easily be verified that $\Psi\geq 0$; and, if $X$ and $Y$ are quasi-independent, then $\Psi=0$. For $c_0$-universal kernels \cite{sriperumbudur2011universality}, we  have that $\Psi=0$ if and only if $X$ and $Y$ are quasi-independent: see Theorem \ref{theo:consistency}.

Given the i.i.d. sample $((X_i,Y_i))_{i\in[n]}$, we can estimate $\Psi$ via $\Psi_n$, defined as
\begin{align}
\Psi_n&=\sup_{\omega\in B_1(\mathcal{H})}\left(\frac{1}{n}\sum_{i=1}^n\omega(X_i,Y_i)\widehat\pi(X_i,Y_i)-\frac{1}{n^2}\sum_{i=1}^n\sum_{k=1}^n\omega(X_i,Y_k)\ind_{\{X_k\leq X_i<Y_k\leq Y_i\}}\right)\label{eqn:Test2}
\end{align}
where $\widehat{\pi}(x,y)=\frac{1}{n}\sum_{m=1}^n\ind_{\{X_m\leq x,Y_m\geq y\}}$, and notice that $\widehat\pi(x,y)$ estimates $\pi(x,y)$. Using a reproducing kernel $\mathfrak{K}$ that factorises, we obtain a simple expression for $\Psi_n^2$:

\begin{proposition}\label{Prop:evaluationQuasiAlte}
 Consider $\mathfrak{K}((x,y),(x',y'))=K(x,x')L(y,y')$. Then
\begin{align*}
\Psi_n^2
&=\frac{1}{n^2}\text{trace}(\boldsymbol{K}\boldsymbol{\widehat{\pi}}\boldsymbol{L}\boldsymbol{\widehat{\pi}}-2{\boldsymbol K}\boldsymbol{\widehat{\pi}}{\boldsymbol L}\boldsymbol{{A}}^\intercal+\boldsymbol{K}\boldsymbol{ A}\boldsymbol{L}\boldsymbol{{A}}^\intercal)
\end{align*}
where $\boldsymbol K$, $\boldsymbol L$, and $\boldsymbol A$ are $n\times n$-matrices with entries given by $\boldsymbol{K}_{ik}=K(X_i,X_k)$, $\boldsymbol{L}_{ik}=L(Y_i,Y_k)$ and $\boldsymbol{{A}}_{ik}=\ind_{\{X_k\leq X_i< Y_k\leq Y_i\}}/n$, and $\boldsymbol{\widehat{\pi}}$ is a diagonal matrix with entries $\boldsymbol{\widehat{\pi}}_{ii}=\widehat{\pi}(X_i,Y_i)$.
\end{proposition}
We remark that the previous expression is similar in form to the Hilbert Schmidt Independence Criterion \cite{GreBouSmoSch05}. In particular, for empirical distributions, $\text{HSIC}(\widehat F_{XY},\widehat F_{X}\widehat F_Y)=\frac{1}{n^2}\text{trace}(\boldsymbol{K} \boldsymbol H^\intercal \boldsymbol L \boldsymbol H)$ with $\boldsymbol H=\boldsymbol I_n-\frac{1}{n}\boldsymbol 1_n\boldsymbol 1_n^\intercal$, whereas our test-statistic can be rewritten as $\Psi_n^2=\frac{1}{n^2}\text{trace}(\boldsymbol K \boldsymbol {\tilde H}^\intercal \boldsymbol L \boldsymbol{\tilde H})$ with $\boldsymbol{\tilde H}=(\boldsymbol{\widehat \pi}-\boldsymbol{A}^\intercal)$. Note that $\boldsymbol{\tilde H}$ is much more complex than $\boldsymbol H$, being a random matrix where each entry depends on all the data points. As we will see, this issue makes the asymptotic analysis in our case much more challenging; by contrast, the asymptotic distribution for HSIC can be readily obtained using standard  results on U-statistics  \cite{gretton2008kernel,ChwGre14}.

Our test can be understood as a generalisation of the log-rank test proposed by \cite{EmuraWang2010}, where instead of considering a single log-rank test with a specific weight function, we consider the supremum over a collection of log-rank tests with weights in $B_1(\mathcal{H})$. By choosing
a sufficiently rich RKHS, for example the RKHS induced by the exponentiated quadratic kernels, we are able to ensure power against a broad family of alternatives. Conversely, simple kernels can recover classical  parametric tests such as the aforementioned log-rank tests. As explained by \cite[Equation 7]{EmuraWang2010}, the simplest possible (constant) function space recovers the well-known conditional Kendall's tau:

\begin{proposition}[Recovering conditional Kendall's tau]
Consider $\mathfrak{K}=1$,  then $\Psi_n^2=K_a^2/n^2$, where $
K_a=\sum_{i<k}\ind_{\{X_i\vee X_k \leq Y_i\wedge Y_k\}}\textnormal{sign} \left((X_i-X_k)(Y_i-Y_k)\right)$ is an empirical estimator of the conditional Kendall's tau.
\end{proposition}

\section{Right-censoring}\label{sec:right-censoring}
In clinical trails, for example, patients might withdraw from the study before observing the time $Y$ of interest leading to so-called right-censored data. To model this kind of data, we introduce additionally the random censoring time $C$. The  data correspond now to i.i.d. samples $((X_i,T_i,\Delta_i))_{i\in[n]}$, where $T_i=\min\{Y_i,C_i\}$ is the observation time, and $\Delta_i=\ind_{\{T_i=Y_i\}}$ is the corresponding censoring status. In particular, if $\Delta_i=0$, we only observe the censoring time $T_i=C_i$, and not the time of interest $Y_i$. Throughout, we assume that $X_i<T_i$ always holds, to reflect the natural ordering of the times, i.e. first recruitment and second the event of interest or the withdrawal from the study. As for the uncensored setting, $X$, $Y$ and $C$ are supposed to be continuously distributed on $\R_+$. Our results are valid under the standard non-informative censoring assumption:
\begin{assumption}\label{Assu:1}
The censoring times are independent of the survival times given the entry times, i.e., $C_i\perp Y_i|X_i$.
\end{assumption}
Standard notation for marginal, joint and conditional densities will be used: for instance, $f_C$, $f_{XT}$ and $f_{Y|X=x}$, are the marginal density of $C$, the joint density of $X$ and $T$, and the conditional density of $Y$ given $X=x$, respectively. Moreover, $S_Y$ denotes the survival function of $Y$, defined as $S_{Y}(y)=\Prob(Y\geq y)$ and $S_{C|X=x}(y)=\Prob(Y\geq y|X=x)$ is the conditional survival function of $Y$ given $X=x$. Under Assumption \ref{Assu:1} we have $S_{T|X=x}(y)=S_{Y|X=x}(y)S_{C|X=x}(y)$.

The null hypothesis of \emph{quasi-independence} is formulated, for the right-censored setting, as 
\begin{align}\label{eqn:factorize_den}
H_0:f_{XY}(x,y)=\tilde f_X(x)\tilde f_Y(y),\quad \text{for all }x \leq y, \text{ s.t. }S_{T|X=x}(y)>0.
\end{align}
As with the uncensored case, $\tilde{f}_X$ and $\tilde f_Y$ are not necessarily equal to the marginal densities $f_X$ and $f_Y$. The additional condition $S_{T|X=x}(y)>0$ ensures that the pair $(x,y)$ is actually observable despite the censoring. The statistic $\Psi$ from Equation \eqref{eq:kqic} is then extended to the  censored setting,
\begin{align*}
\Psi_c &=\sup_{\omega\in B_1(\mathcal{H})}\int_{x\leq y}\omega(x,y)\rho^c(x,y)dxdy \geq 0, \\
\text{where }\,\rho^c(x,y) & =-\pi^c(x,y)\frac{\partial^2}{\partial x\partial y}\pi^c_1(x,y)+\frac{\partial\pi^c(x,y)}{\partial x} \frac{\partial\pi_1^c(x,y)}{\partial y},
\end{align*}
and $\pi^{c}_1(x,y)=\Prob(X\leq x, T\geq y,\Delta=1)$ and $\pi^{c}(x,y)=\Prob(X\leq x, T\geq y)$ for $x \leq y$. 

\begin{proposition}\label{prop:psi_c=0}
We have $\Psi_c= 0$ if the null hypothesis $H_0$ of quasi-independence is fulfilled.
\end{proposition}
The (updated) estimator for the (new) Kernel Quasi Independent Criterion $\Psi_c$ is defined by
\begin{align}
\Psi_{c,n}&= \sup_{\omega\in B_1(\mathcal{H})}\left(\frac{1}{n}\sum_{i=1}^n\Delta_i\omega(X_i,T_i)\widehat\pi^c(X_i,T_i)-\frac{1}{n^2}\sum_{i=1}^n\sum_{k=1}^n\Delta_k\omega(X_i,T_k)\ind_{\{X_k\leq X_i<T_k\leq T_i\}}\right),\label{eqn:Test2cen}
\end{align}
where $\widehat{\pi}^c(x,y)=\frac{1}{n}\sum_{m=1}^n\ind_{\{X_m\leq x,T_m\geq y\}}$ is the natural estimator for $\pi^c(x,y)$. In the uncensored case, i.e. $\Delta=1$ with probability 1, the new KQIC $\Psi_c$ and its estimator $\Psi_{c,n}$ collapse to the respective quantities $\Psi$ and $\Psi_n$ from Section~\ref{Sec:Quasi}. Moreover, the estimator $\Psi_{c,n}$ can be simplified for factorising kernels:

\begin{proposition}\label{Prop:Cen-evaluationQuasi}
Consider $\mathfrak{K}((x,y),(x',y'))=K(x,x')L(y,y')$, then
\begin{align}\label{eqn:psi_2}
 \Psi_{c,n}^2
&=\frac{1}{n^2}\text{trace}(\boldsymbol{K}\boldsymbol{\widehat{\pi}}^c\boldsymbol{\tilde L}\boldsymbol{\widehat{\pi}}^c-2{\boldsymbol K}\boldsymbol{\widehat{\pi}}^c{\boldsymbol{\tilde L}}\boldsymbol{{B}}^\intercal+\boldsymbol{K}\boldsymbol{ B}\boldsymbol{{ L}}\boldsymbol{{B}}^\intercal)
\end{align}
where $\boldsymbol{K}_{ik}=K(X_i,X_k)$, $\boldsymbol{\tilde L}_{ik}=\Delta_i\Delta_k L(T_i,T_k)$, $\boldsymbol{{B}}_{ik}=\ind_{\{X_k\leq X_i< T_k\leq T_i\}}/n$, and $\boldsymbol\pi^c$ is a diagonal matrix where $\boldsymbol{\widehat{\pi}}^c_{ii}=\widehat{\pi}(X_i,T_i)$.
\end{proposition}

\section{Asymptotic analysis and wild bootstrap}\label{Sec:asymptotic}

We now present our main two theoretical results. First, we establish the asymptotic null distribution of our statistic $n\Psi^2_{c,n}$.
\begin{theorem}\label{thm:null}
Assume $\mathfrak K$ is bounded. Then, under the null hypothesis, $n\Psi_{c,n}^2\overset{\mathcal D}{\to}\mu+\mathcal Y$, where $\mu$ is a positive constant, $\mathcal Y=\sum_{i=1}^\infty\lambda_i(\xi^2_i-1)$, $\xi_1,\xi_2,\ldots$ are independent standard normal random variables, and $\lambda_1,\lambda_2,\ldots$ are non-negative constants depending on the distribution of the random variables $(X,Y,C)$ and the kernel $\mathfrak{K}$. 
\end{theorem}

To verify Theorem \ref{thm:null}, we show that the scaled version of our statistic, $n\Psi_{c,n}^2$, can be expressed under the null hypothesis as the sum of a certain V-statistic and an asymptotically vanishing term. To find this representation, we write our test-statistic as a double integral with respect to a martingale, and use martingale techniques, and the results introduced in \cite{FerRiv2020}, to show that the error incurred by replacing certain quantities by their population versions vanishes as the number of data points grows to infinity.  The full proof is provided in Appendix \ref{sec:proof null}. We next establish conditions for consistency of the test under the alternative.
\begin{theorem}\label{theo:consistency}
    Let $\mathfrak K$ be a bounded, $c_0$-universal kernel \cite{sriperumbudur2011universality}. Then $\Psi_{c,n}^2 \to \Psi_{c}^2$ in probability. Moreover,  whenever the null hypothesis is violated, $\Psi_c^2$ is positive, implying that $n\Psi_{c,n}^2\to \infty$ in probability,. 
\end{theorem}
We remark that the factorised kernel $\mathfrak K ((x,y),(x',y'))=K(x,x')L(y,y')$ must be $c_0$-universal in the product space, which is true for instance when $K$ and $L$ are exponentiated quadratic kernels \cite{FukGreSunSch08}.  In the case of independence testing, a simpler condition on the kernel can be used, where kernels are required to be individually characteristic to their respective domains \cite{Gretton15}. Whether this simple condition can be generalised to the quasi-independence setting remains a topic for future work.

The consistency result in Theorem \ref{theo:consistency} relies on the interpretation of the test statistic $\Psi_{c,n}$ and the KQIC $\Psi_c$, as the Hilbert space distances of the embeddings of certain positive measures. These distances measure the degree of (quasi)-dependence. Under the $c_0$-universality assumption, the embedding of finite signed measures are injective  \cite{sriperumbudur2011universality}, which, in our case,  implies $\rho^c(x,y)=0$ for almost all $x\leq y$.  It remains to prove that quasi-independence holds. To show this, we first note that  $\rho^c(x,y)=0$ implies
\begin{align}
\frac{\partial^2\pi^c_1(x,y)}{\partial x\partial y}=\frac{1}{\pi^c(x,y)}\frac{\partial\pi^c(x,y)}{\partial x} \frac{\partial\pi_1^c(x,y)}{\partial y},\label{eqn:id}
\end{align}
and that $\frac{\partial^2\pi^c_1(x,y)}{\partial x\partial y} =S_{C|X=x}(y)f_{XY}(x,y)$. By carefully analysing Equation \eqref{eqn:id} we find an explicit decomposition of $f_{XY}(x,y)$ into the product of two functions only depending on $x$ and $y$, respectively, from which quasi-independence follows. A detailed proof is provided in Appendix \ref{sec:proofconsistency}.

As noted above, the eigenvalues $\lambda_i$ in Theorem \ref{thm:null}  --- and thus, the limit distribution of our test statistic under the null hypothesis --- depend on the unknown distribution of  $(X,Y,C)$. For this reason, we propose to approximate the limit null distribution and its $(1-\alpha)$-quantile $q_{\alpha}$ of $\mu + \mathcal Y$ using a  wild bootstrap approach. This strategy is well-established for $V$- and $U$-statistics \cite{dehling1994random}, and has  successfully been applied in scenarios, similar to the present one, where the test statistic  behaves asymptotically as a $V$-statistic \cite{fernandez2019kernel, fernandez2019reproducing}. 

To introduce the wild bootstrap counterpart $\Psi_{c,n}^{\text{WB}}$ of our statistic $\Psi_{c,n}$, let $W_1,\ldots, W_n$ be  independent and identically distributed Rademacher random variables, and define the $n\times n$ matrix $\boldsymbol{K}^W$ with entries $\boldsymbol{K}^W_{ik}=W_iW_kK(X_i,X_k)$. Then, 
\begin{align*}
(\Psi_{c,n}^{\text{WB}})^2=\frac{1}{n^2}\text{trace}(\boldsymbol{K}^W\boldsymbol{\widehat{\pi}}^c\boldsymbol{\tilde L}\boldsymbol{\widehat{\pi}}^c-2{\boldsymbol K}^W\boldsymbol{\widehat{\pi}}^c{\boldsymbol {\tilde L}}\boldsymbol{{B}}^\intercal+ \boldsymbol{K}^W\boldsymbol{ B}\boldsymbol{\tilde L}\boldsymbol{{B}}^\intercal).
\end{align*}
We propose the test $\varphi_n^{\text{WB}} = \ind\{\Psi_{c,n}^2 > q_\alpha^{\text{WB}}\}$ to infer $H_0$, where $q_\alpha^{\text{WB}}$ denotes the $(1-\alpha)$-quantile of $ \Psi_{c,n}^{\text{WB}})^2$ given the observations $((X_i,\Delta_i,T_i))_{i\in[n]}$.

\section{Experiments}\label{sec:Experiments}
We perform synthetic experiments followed by real data applications. In the first set of synthetic examples, we replicate the settings studied in \cite{chiou2018permutation}, where  Gaussian copula models were used to create dependencies between $X$ and $Y$. In the second synthetic experiment, we investigate distribution functions $f_{Y|X=x}$ that have a periodic dependence on $x$. We then apply our tests to real-data scenarios such as those studied in \cite{EmuraWang2010} and \cite{meister2008statistical}.

{\bf Methods}
We implement the proposed quasi-independence test based on the test-statistic \textbf{KQIC} given in Equation \eqref{eqn:psi_2}. The kernels are chosen to be Gaussian with bandwidth optimised by using approximate test power \cite{gretton2009fast,jitkrittum2017linear}. See  Appendix \ref{app:experiments} for details. Competing approaches include: \textbf{WLR}, the weighted log-rank test proposed in \cite{EmuraWang2010}, with weight function chosen equal to $n\widehat\pi^c(x,y)$;\footnote{Our test-statistic recovers, as a particular case, the squared of this log-rank test by choosing $\mathfrak K=1$}  \textbf{WLR\_SC}, the weighted log-rank test proposed in \cite{EmuraWang2010}, with weight function chosen as suggested by the authors, i.e,  $W(x,y)=\int_0^{x} \widehat{S}_{C_R}((y-u)-)^{-1}\widehat\pi^c(du,y)$, where $\widehat S_{C_R}$ is the Kaplan-Meier estimator associated to the data $((C_i-X_i,1-\Delta_i))_{i=1}^n$; \textbf{M\&B}, the conditional Kendall's tau statistic modified to incorporate censoring as proposed in \cite{MartinBetensky2005}; and \textbf{MinP1} and \textbf{MinP2}, the ``minimal p-value selection" tests proposed in \cite{chiou2018permutation}, which rely on permutations of the observed pairs. A review of these approaches can be found in Appendix \ref{sec:method_review}. For the synthetic experiments, we recorded the rejection rate over 200 trials. The wild-bootstrap size for \textbf{KQIC} and the permutation size for \textbf{MinP1}, \textbf{MinP2} are set to be 500.

{\bf Monotonic Dependency}
The first synthetic example from \cite{chiou2018permutation} is generated as follows: $X \sim \textnormal{Exp}(5)$ and $Y \sim \textnormal{Weibull}(3, 8.5)$; $(X, Y)$ are then coupled via a 2-dimensional Gaussian copula model with correlation parameter $\rho$. The censoring variable is set to be exponentially distributed and truncation applies. With the copula construction, the magnitude of the correlation parameter $\rho$ is a fair indicator of the degree of dependence, with $\rho = 0$ denoting independence. Rejection rates are reported in Table \ref{tab:syn_exp1_monotone}. At $\rho = 0$, the null hypothesis holds, and the rejection rates refer to the Type-I error. All the tests achieve a correct Type-I error around a test level $\alpha = 0.05$. For $\rho \neq 0$, the alternative holds, and the rejection rates correspond to test power (the higher the better). The highest value is in bold. 
Test results w.r.t. different censoring rates can be found in the  Appendix. Overall, our method outperforms all competing approaches.  

\begin{table}[ht]
    \centering
\begin{tabular}{lrrrrr}
\toprule
{\quad $\rho$} &  -0.4 &  -0.2 &  0.0 &   0.2 &   0.4 \\
\midrule
KQIC   &  \tb{0.93} &  \tb{0.46} &  0.06 &  \tb{0.42} &  \tb{0.86} \\
WLR     &  0.80 &  0.33 &  0.10 &  0.18 &  0.66 \\
WLR\_SC &  0.85 &  {0.42}  &  0.03 &  0.24 & 0.74 \\
M\&B &  0.64 &  0.22 &  0.02 &  0.16 &  0.74 \\
MinP1   &  0.58 &  0.12 &  0.03 &  0.17 &  0.62 \\
MinP2   &  0.33 &  0.04 &  0.06 &  0.10 &  0.28 \\
\bottomrule
\end{tabular}
\begin{tabular}{lrrrrr}
\toprule
{} &  -0.4 &  -0.2 &  0.0 &   0.2 &   0.4 \\
\midrule
{} &  \tb{0.99} &  \tb{0.67} &  0.05 &  \tb{0.63} &   \tb{1.00} \\
{} &  0.94 &  0.52 &  0.06 &  0.32 &  0.94 \\
{} &  0.93 &  0.53 &  0.06 &  0.43 &  0.99\\
{} &  0.94 &  0.28 &  0.03 &  0.42 &  0.92 \\
{} &  0.84 &  0.12 &  0.10 &  0.34 &  0.84 \\
{} &  0.56 &  0.08 &  0.08 &  0.28 &  0.52 \\
\bottomrule
\end{tabular}
\vspace{0.3cm}
\caption{ Rejection rates for monotonic dependency models based on Gaussian copula, with $n=100$ on the left; $n=200$ on the right; $\alpha = 0.05$; censoring rate: $50\%$.
}\label{tab:syn_exp1_monotone}
\end{table}

{\bf V-shaped Dependency}
A synthetic example \cite{chiou2018permutation}, in which the authors compare the behaviour of their tests against the conditional Kendall's tau test of \cite{MartinBetensky2005} in detecting non-monotonic dependencies. The following V-shaped dependency structure applies: $X \sim \textnormal{Weibull}(0.5, 4)$; $Y \sim \textnormal{Uniform}[0,1]$; $(X, |Y-0.5|)$ is coupled via the 2-dimensional Gaussian copula with correlation coefficient $\rho$ as above. Exponential censoring and truncation apply. Rejection rates are plotted against the perturbation of correlation coefficient $\rho$ in Figure \ref{fig:vshape_gaussian_copula_res}, where  {KQIC} outperforms competing methods.

\begin{figure}[t]
\centering
\includegraphics[scale=0.313]{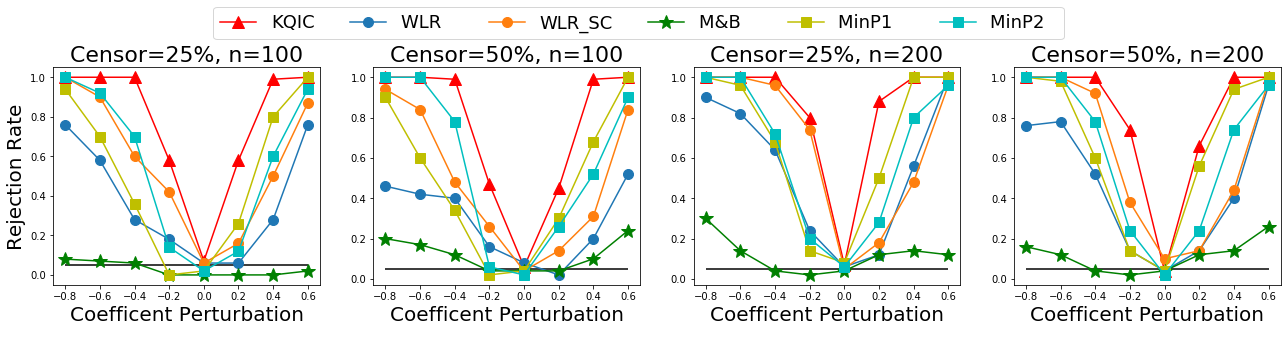}
    \caption{Rejection rate for V-shape Gaussian copula model}
    \label{fig:vshape_gaussian_copula_res}
\end{figure}

\begin{figure}[t!]
\centering
    \includegraphics[scale=0.313]{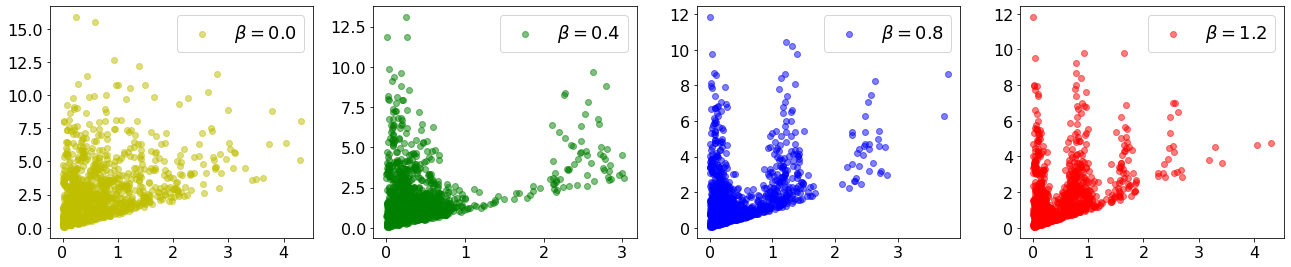}
    \caption{Samples from Periodic Dependency Model w.r.t. Frequency Coefficient $\beta$}
    \label{fig:periodic_sample}
\end{figure}
\begin{figure}[t!]
\centering
\includegraphics[scale=0.313]{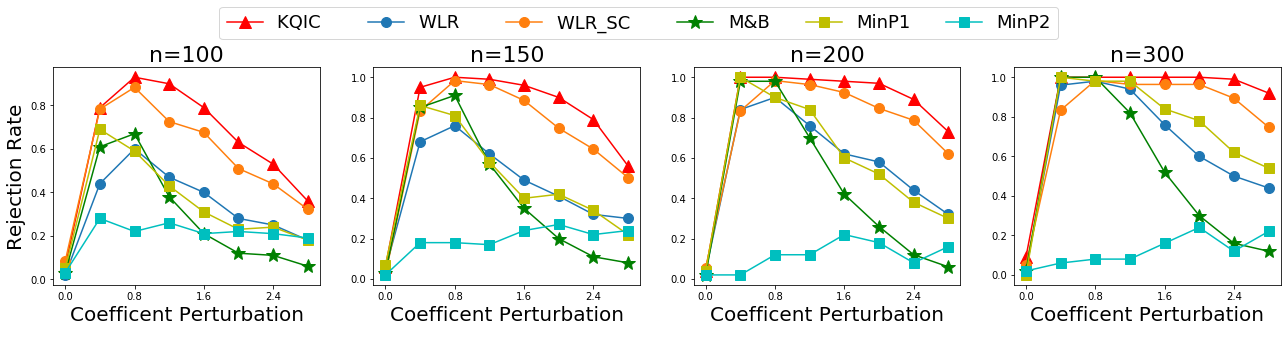}
    \caption{Rejection Rate for Periodic Dependency Model with $25\%$ data censored.}
    \label{fig:periodic_res}
\end{figure}

{\bf Periodic Dependency}
Apart from the V-shaped dependencies studied in \cite{chiou2018permutation}, we investigate more complicated non-monotonic dependencies structures. The data are generated with a periodic dependency structure, $X \sim \textnormal{Exp}(1)$; $Y|X \sim \textnormal{Exp}(e^{\cos (2 \pi \beta  X})$. The coefficient $\beta$ controls the frequency of the dependence. A set of examples with different parameters $\beta$ is shown in Figure \ref{fig:periodic_sample}, with $\beta = 0$ implying independence. Further details are discussed in Appendix \ref{app:experiments_period}.

Examining the results in Figure \ref{fig:periodic_res}, we see that our method outperforms competing approaches.
Unlike the correlation coefficient $\rho$ in Gaussian copula models, the coefficient $\beta$ does not directly imply the ``amount'' of dependence; rather, a higher $\beta$ indicates a more ``difficult'' problem. Thus, as anticipated, power  
drops for large values of $\beta$, and the effect is more apparent at low sample sizes. Note in particular that the permutation based tests \cite{chiou2018permutation} are more affected by an increase in frequency at which dependence occurs, while our test shows a more robust behaviour.

{\bf High Frequency Dependency}
In the period dependency problem above, the parameter $\beta$ controls the frequency of sinusoidal dependence. At a given  sample size, the dependence  becomes harder to detect as the frequency $\beta$ increases. We visually show this in Appendix \ref{app:experiments_period}. For problems with high frequency dependence, a larger sample size is required. 

When the sample size increases, KQIC is able to successfully reject the null at relatively high frequencies (large $\beta$), as shown in Figure \ref{fig:periodic_res_supp}. At lower frequencies $\beta = 3.0$, WLR\_SC has similar test power as KQIC. As the problem gets  harder with larger $\beta$, KQIC outperforms WLR\_SC. The IMQ kernel has similar test power as the Gaussian kernel on this example.
We report the Type-I error that is well controlled in Appendix \ref{app:experiments_period} 
Table \ref{tab:periodic_h0}.

\begin{figure}[t!]
\centering  
    \includegraphics[width=1.0\textwidth]{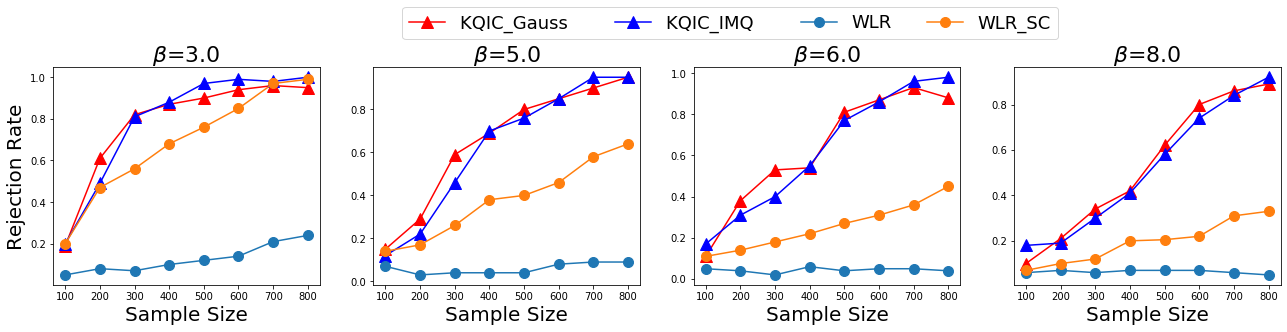}
    \caption{Rejection rate for high frequency dependency, with  $\alpha=0.05$, $40\%$ data censored}\label{fig:periodic_res_supp}
\end{figure}
\begin{figure}[t!]
    \centering
    \includegraphics[width=1.0\textwidth]{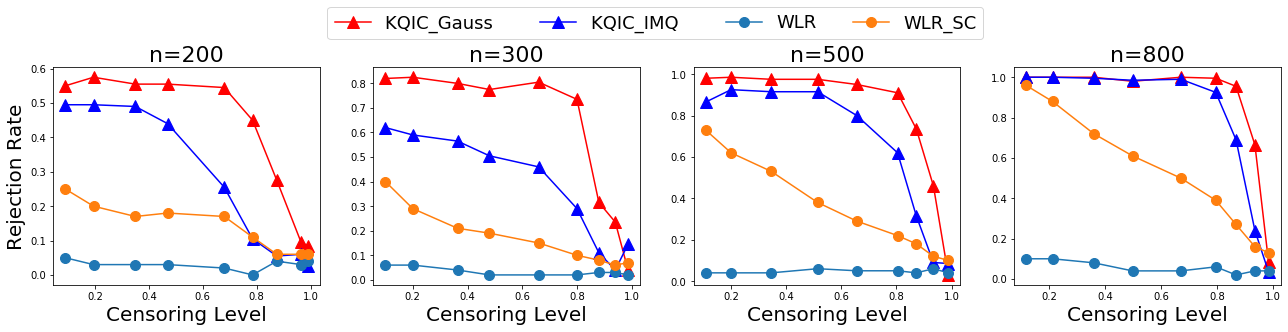}
    \caption{Rejection rate for periodic dependencies ($\beta=5.0$), with $\alpha=0.05$ and $200$ trials.}
    \label{fig:censor_level}
\end{figure}

{\bf Censoring level}
We investigate how our test is affected by the censoring level, in particular when the censoring percentage increases. We analyse performance under both the null and alternative hypotheses. The Type-I error is well controlled for KQIC and details are reported in Appendix \ref{app:censor_level}.

Under the alternative hypothesis, in Figure \ref{fig:censor_level}, we show the rejection rate w.r.t. different censoring percentages and fixed sample size. This is done in our periodic dependency setting. From the plot, we see that KQIC with Gaussian and IMQ kernels is more robust to censoring, with test power starting to drop at $85\%$ of censoring for sample size = 800.  WLR\_SC is strongly affected by censoring. WLR is not capable of detecting $H_1$ in this hard problem with high frequency. 

In addition, we study the test behaviour with dependent censoring, since in Assumption \ref{Assu:1}, only conditional independence $Y\perp C |X$ is required \cite{EmuraWang2010}. Detailed results are reported in Appendix \ref{app:dependent censoring}.

{\bf Computational cost}
Our proposed test, implemented as described in Appendix \ref{sec:efficient}, has a significantly lower runtime when compared with the competing permutation approaches. M\&B implements the conditional Kendall's tau statistic, which has a closed-form expression for the null distribution, therefore its runtime is lowest of all. See  Appendix \ref{app:experiments} for details.

{\bf Real Data Experiment} We consider three real data scenarios: \tb{Channing House} \cite{hyde1977testing}: contains  the recorded entry times and lifetimes of 461 patients ($97$ men and $364$ women). Among them, $268$ subjects withdrew from the retirement center, yielding to a censoring proportion of $0.62$. The data are naturally left truncated, as only patients who entered the center are observed; \tb{AIDS} \cite{LagakosBarrajDeGruttola1988}: the data contain the incubation time and lapse time, measured from infection to recruitment time,  for $295$ subjects. A censoring of proportion of $0.125$ occurs due to death or withdrawal from the study. Left truncation applies since only patients that developed AIDS within the study period were recruited, thus only patients with incubation time not exceeding the lapse time were observed; and \tb{Abortion} \cite{meister2008statistical}: contains the entry time and the spontaneous abortion time for $1186$ women ($197$ control group and $989$ treatment group exposed to Coumarin derivatives).  A censoring proportion of $0.906$ occurs due to live birth or induced abortions.  Delayed entry to the study is substantial in this dataset:  $50\%$ of the control cohort entered the study in week 9 or later, while in the treatment group this occurs for $25\%$ of the cohort. 
\tb{Implementation:} For our test we used both Gaussian kernels KQIC\_Gauss and IMQ kernels KQIC\_IMQ. For competing approaches, the implementation is as discussed at the beginning of this section. \tb{Results}:
For the Channing house dataset, in Table \ref{tab:real_data},  we observe that all tests agree in not rejecting the null hypothesis for the combined and female groups at a level $\alpha=0.05$. For the male group, all tests but MinP2 and M\&B reject the null hypothesis at $\alpha=0,05$. Our results agree with \cite{EmuraWang2010}. For the AIDS dataset, all tests reach a consensus of rejecting the null, which is consistent with \cite{EmuraWang2010}, except for {MinP2} marked in blue.  For the abortion dataset, our test rejects the null hypothesis, suggesting dependency between the entry time $X$ and the spontaneous abortion time $Y$ in both the treatment group and the combined case (in red). This finding is in accordance with domain knowledge \cite{meister2008statistical}, where the presence of this dependence was indicated to be due to the study design. The competing tests were unable to detect the dependence; however, did not reject the null hypothesis.

\begin{table}[]
    \centering
\begin{tabular}{l|ccc|c|ccc}
\toprule
\multirow{2}{2cm}{\quad (p-value)}&
    \multicolumn{3}{c|}{Channing House}&\multicolumn{1}{c|}{AIDS}&\multicolumn{3}{c}{Abortion Times}
    \\
{} &  Combined &   Male  &  Female &   {} &  Combined &  Control &  Treatment \\
\midrule
KQIC\_Gauss &     0.072 &  0.012 &   0.566 &  0.030 &  \color{red}0.014 &    0.440 &      \color{red}0.028 \\
KQIC\_IMQ   &     0.078 &  0.022 &   0.414 &  0.010 &  \color{red}0.032 &    0.158 &      \color{red}  0.048 \\
WLR        &     0.058 &  0.016 &   0.444 &  0.035 &    0.408 &    0.868 &       0.748 \\
WLR\_SC &  0.086 &  0.020 &  0.422 &  0.030 &  0.511 &  0.674 &  0.450 \\

MinP1 &  0.084 &  0.036 &  0.396 &  0.012 &  0.584 &  0.584 &  0.452 \\
MinP2 &  0.198 &  \color{blue}0.426 &  0.118 & \color{blue}  0.406 &  0.694 &  0.572 &  0.346 \\
M\&B &  0.178 &  \color{blue}0.199 &  0.495 &  0.010 &  0.712 &  0.693 &  0.752 \\
\midrule
$\%$ Events &  0.379 &  0.474 &  0.354 &  0.875 &  0.094 &  0.069 &  0.098 \\
\bottomrule
\end{tabular}
\vspace{0.3cm}
    \caption{Real data, with marked results $\color{blue}{\text{contradicting}}$ and $\color{red}{\text{supporting}}$ the scientific literature.}
    \label{tab:real_data}
\end{table}

\section{Conclusions}
We address the problem of testing for quasi-independence in the presence of left-truncation, as occurs in real-world examples where events are ordered. The test is nonparametric and general-purpose, can detect a broad class of departures from the null, and applies even where right-censoring is present.  In experiments on challenging synthetic data, our method strongly outperforms the alternatives. On real-life datasets, our method yields consistent results to classical approaches where these apply; however, it also detects quasi-dependence in a case where  competing  general-purpose approaches fail, and where models based on domain knowledge were needed in establishing the result.
Our tests are a first step towards the wider challenge of testing quasi-independence in the presence of general  physical or causal constraints on the variables, which themselves induces a ``baseline'' level of dependence. Many real-world settings, apart from left-truncation, do not have the advantage of a ``pure'' null scenario of perfect independence, and tests must be designed in light of these constraints. These directions are an exciting topic for future study.

\newpage
\section{Broader Impact}
\paragraph{Potential benefits to society}
Finding dependencies is a key tool in a broad variety of scientific domains, including clinical treatments, demography, business strategy development, and public policy formulation, with applications spanning the natural and social sciences. Our work addresses these questions by studying the dependence relationships between observed data, where the data already have an intrinsic dependence due to natural order. Moreover, the dependence need not be monotonic, but can take a variety of forms.
Detecting the dependence of variables in this setting, which corresponds to many real-life scenarios,  will allow  scientists or policymakers to better understand their data and research problems, and guide the better design of future research questions.
The dependence detection strategy may be used to detect bias in data collection procedures. Such bias could be avoided by verifying the  absence of inadvertent dependency relationships in collected data.

\paragraph{Potential risks to society}
There are a number of ways in which statistical tests can be mis-applied in the wider scientific community, and these must be guarded against. As one example, p-value hacking/failure to correct for multiple testing can result in false positives. In the event that these false positives are surprising or controversial, they can gain considerable traction in the media. In some cases, peoples' health can be at risk.
A second risk, specific to tests of dependence, is for correlation and causation to be confused. Our tests detect correlation, however, a misunderstanding of such tests might result in false conclusions of cause and effect. There have been especially pernicious instances when using statistics in domains such as crime prediction. 

\section*{Acknowledgement}
TF, WX, and AG thank the Gatsby Charitable Foundation for the financial support. 
MD gratefully acknowledge support from the Deutsche Forschungsgemeinschaft (grant no. PA-2409 5-1).

\bibliographystyle{plain}
\bibliography{ref}

\newpage

\appendix
\paragraph{\large{Appendix: A kernel test for quasi-independence}}

\begin{itemize}
    \item[]\textbf{Appendix A:} Preliminary results
    \item[]\textbf{Appendix B:} Proofs of sections \ref{Sec:Quasi} and \ref{sec:right-censoring}
    \item[]\textbf{Appendix C:} Proof of Theorem \ref{thm:null} (null distribution)
    \item[]\textbf{Appendix D:} Proof of Theorem \ref{theo:consistency} (consistency under alternatives)
    \item[]\textbf{Appendix E:} Efficient wild bootstrap implementation
    \item[]\textbf{Appendix F:} Review of related quasi-independence tests
    \item[]\textbf{Appendix G:} Additional experiments and discussion
\end{itemize}

\section{Preliminary results}
The following Proposition is an intermediate result, which is need to prove Lemmas \ref{lemma:approx} and \ref{lem:conv_embeddings}.
\begin{proposition}\label{prop:sumAsmall}
	Define
	\begin{align}
	A_n=\frac{1}{n}\sum_{i=1}^n\Delta_i \ind_{\{X_i< T_i\}}(\widehat{\pi}^c(X_i,T_i)- {\pi}^c(X_i,T_i))^2.
	\end{align}
    Then, the following results hold: i) $A_n\to 0$ almost surely as $n$ grows to infinity and ii) $|A_n|\leq c$, for some constant $c$, for all large $n$.
\end{proposition}
\begin{proof}
	Since $\widehat{\pi}^c$ and ${\pi}^c$ are both bounded by 1, we have $|A_n|=A_n\leq 4$ for all $n$ and, thus, ii) is proven.
	
	Let us consider the statement i). It is easy to see that $\E(\ind_{\{X_m\leq x,T_m\geq t\}}) = \pi^c(x,t)$. In particular, we have $\E(g(m,i)|X_i,T_i,\Delta_i)=0$ for $i\neq m$, where $g(m,i)=\ind_{\{X_m\leq X_i,T_m\geq T_i\}}-{\pi}^c(X_i,T_i)$. Now, notice that we can rewrite $A_n$ as $V$-statistic of order 3:
	\begin{align*}
		A_n&=\frac{1}{n}\sum_{i=1}^n\Delta_i\ind_{\{X_i\leq T_i\}}\left(\frac{1}{n}\sum_{m=1}^n(\ind_{\{X_m\leq X_i,T_m\geq T_i\}}- {\pi}^c(X_i,T_i))\right)^2\\
		&=\frac{1}{n^3}\sum_{i=1}^n\sum_{m=1}^n\sum_{k=1}^n\Delta_i\ind_{\{X_i\leq T_i\}}g(m,i)g(k,i).
	\end{align*}
	Combining this and the law of large numbers for $V$-statistics yields 
	\begin{align*}
		A_n\overset{a.s.}{\to} \E(\Delta_1 g(1,2)g(1,3))&=\E(\Delta_1\E(g(2,1)g(3,1)|X_1,T_1,\Delta_1))\\
		\text{(independence)}
		&=\E\left(\Delta_1\E(g(2,1)|X_1,T_1,\Delta_1)\E(g(3,1)|X_1,T_1,\Delta_1)\right)\\
		&=0.
	\end{align*} 	
\end{proof}
\section{Proofs of sections \ref{Sec:Quasi} and \ref{sec:right-censoring}}
\subsection{Proof of Proposition \ref{Prop:evaluationQuasiAlte}}

\begin{proof}
From Equation \eqref{eqn:Test2}, we have
\begin{align*}
\Psi_n
&=\sup_{\omega\in B(\mathcal{H})}\frac{1}{n}\sum_{i=1}^n\left(\omega(X_i,Y_i)\widehat{\pi}(X_i,Y_i)-\sum_{k=1}^n\omega(X_i,Y_k)\boldsymbol{ A}_{ik}\right),    
\end{align*}
where $\boldsymbol{A}_{ik}=\ind_{\{X_k\leq X_i< Y_k\leq Y_i\}}/n$.

The previous result and the reproducing kernel property yield
\begin{align*}
 \Psi_n^2
&=\sup_{\omega\in B_1(\mathcal{H})}\left(\frac{1}{n}\sum_{i=1}^n\left(\omega(X_i,Y_i)\widehat{\pi}(X_i,Y_i)-\sum_{k=1}^n\omega(X_i,Y_k)\boldsymbol{ A}_{ik}\right)\right)^2\\
&=\sup_{\omega\in B_1(\mathcal{H})}\left\langle\omega(\cdot),\frac{1}{n}\sum_{i=1}^nK(X_i,\cdot)\left(L(Y_i,\cdot)\boldsymbol{\widehat{\pi}}_{ii}-\sum_{k=1}^nL(Y_k,\cdot)\boldsymbol{ A}_{ik}\right)\right\rangle^2\\
&=\left\|\frac{1}{n}\sum_{i=1}^nK(X_i,\cdot)\left(L(Y_i,\cdot)\boldsymbol{\widehat{\pi}}_{ii}-\sum_{k=1}^nL(Y_k,\cdot)\boldsymbol{ A}_{ik}\right)\right\|^2_{\mathcal{H}}\\
&=\frac{1}{n^2}\sum_{i,j=1}^n{\boldsymbol K}_{ij}{\boldsymbol L}_{ij}\boldsymbol{\widehat{\pi}}_{ii}\boldsymbol{\widehat{\pi}}_{jj}-\frac{2}{n^2}\sum_{i,j,l=1}^n{\boldsymbol K}_{ij}{\boldsymbol L}_{il}\boldsymbol{\widehat{\pi}}_{ii}\boldsymbol{ A}_{jl}+\frac{1}{n^2}\sum_{i,j,k,l=1}^n{\boldsymbol K}_{ij}{\boldsymbol L}_{kl}\boldsymbol{ A}_{ik}\boldsymbol{ A}_{jl}\\
&=\frac{1}{n^2}\text{trace}(\boldsymbol{K}\boldsymbol{\widehat{\pi}}\boldsymbol{L}\boldsymbol{\widehat{\pi}}-2{\boldsymbol K}\boldsymbol{\widehat{\pi}}{\boldsymbol L}\boldsymbol{{A}}^\intercal+\boldsymbol{K}\boldsymbol{ A}\boldsymbol{L}\boldsymbol{{A}}^\intercal),
\end{align*}
where the second to last equality follows from
\begin{align*}
\frac{1}{n^2}\sum_{i=1}^n\sum_{j=1}^n {\boldsymbol K}_{ij}{\boldsymbol L}_{ij}\boldsymbol{\widehat{\pi}}_{ii}\boldsymbol{\widehat{\pi}}_{jj}&=\frac{1}{n^2}\sum_{i=1}^n\sum_{j=1}^n {\boldsymbol K}_{ij}(\boldsymbol{\widehat{\pi}}{\boldsymbol L}\boldsymbol{\widehat{\pi}})_{ij}=\frac{1}{n^2}\text{trace}(\boldsymbol{K}\boldsymbol{\widehat{\pi}}\boldsymbol{L}\boldsymbol{\widehat{\pi}}),
\end{align*}
\begin{align*}
\frac{2}{n^2}\sum_{i,j,l=1}^n{\boldsymbol K}_{ij}{\boldsymbol L}_{il}\boldsymbol{\widehat{\pi}}_{ii}\boldsymbol{ A}_{jl}
&=\frac{2}{n^2}\sum_{j=1}^n\sum_{l=1}^n\left(\sum_{i=1}^n {\boldsymbol K}_{ij}(\boldsymbol{\widehat{\pi}}{\boldsymbol L})_{il}\right) \boldsymbol{{A}}_{jl}\\
&=\frac{2}{n^2}\sum_{j=1}^n\sum_{l=1}^n\left({\boldsymbol K}\boldsymbol{\widehat{\pi}}{\boldsymbol L}\right)_{jl}\boldsymbol{{A}}^\intercal_{lj}\\
&=\frac{2}{n^2}\text{trace}({\boldsymbol K}\boldsymbol{\widehat{\pi}}{\boldsymbol L}\boldsymbol{{A}}^\intercal),
\end{align*}
and
\begin{align*}
\frac{1}{n^2}\sum_{i,j,k,l=1}^n{\boldsymbol K}_{ij}{\boldsymbol L}_{kl}\boldsymbol{ A}_{ik}\boldsymbol{ A}_{jl}
&=\frac{1}{n^2}\sum_{k=1}^n\sum_{j=1}^n\left(\sum_{i=1}^n {\boldsymbol K}_{ij}\boldsymbol{ A}_{ik}\right) \left(\sum_{l=1}^n {\boldsymbol L}_{kl}\boldsymbol{ A}_{jl}\right)\\
&=\frac{1}{n^2}\sum_{k=1}^n\sum_{j=1}^n\left({\boldsymbol K}\boldsymbol{ A}\right)_{jk}\left(\boldsymbol{L}\boldsymbol{ A}^\intercal\right)_{kj}\\
&=\frac{1}{n^2}\text{trace}(\boldsymbol{K}\boldsymbol{ A}\boldsymbol{L}\boldsymbol{{A}}^\intercal).
\end{align*}
\end{proof}

\subsection{Proof of Proposition \ref{Prop:Cen-evaluationQuasi}}
\begin{proof}
Equation \eqref{eqn:Test2cen} yields
\begin{align*}
\Psi_{c,n}&=\sup_{\omega\in B_1(\mathcal H)}\frac{1}{n}\sum_{i=1}^n\left(\Delta_i\omega(X_i,T_i)\widehat{\pi}^c(X_i,T_i)-\frac{1}{n}\sum_{k=1}^n\Delta_k\omega(X_i,T_k)\ind_{\{X_k\leq X_i\leq T_k\leq T_i\}}\right)\\
&=\sup_{\omega\in B_1(\mathcal H)}\frac{1}{n}\sum_{i=1}^n\left(\Delta_i\omega(X_i,T_i)\boldsymbol{\widehat{\pi}}^c_{ii}-\sum_{k=1}^n\Delta_k\omega(X_i,T_k)\boldsymbol{ B}_{ik}\right),
\end{align*}
where $\boldsymbol{ B}_{ik}=\ind_{\{X_k\leq X_i< T_k\leq T_i\}}/n$ and $\boldsymbol{\widehat{\pi}}^c$ is a diagonal matrix with entries $\boldsymbol{\widehat{\pi}}^c_{ii}=\widehat{\pi}^c(X_i,T_i)$. 

Then, by following the exact same computations of the proof of Proposition \ref{Prop:evaluationQuasiAlte}, we deduce
\begin{align*}
\Psi_{c,n}^2
&=\frac{1}{n^2}\text{trace}(\boldsymbol{K}\boldsymbol{\widehat{\pi}}\boldsymbol{\tilde L}\boldsymbol{\widehat{\pi}}-2{\boldsymbol K}\boldsymbol{\widehat{\pi}}{\boldsymbol {\tilde L}}\boldsymbol{{B}}^\intercal+\boldsymbol{K}\boldsymbol{ B}\boldsymbol{{\tilde L}}\boldsymbol{{B}}^\intercal),
\end{align*}
where $\boldsymbol{\tilde L}_{ik}=\Delta_i\Delta_k L(T_i,T_k)$.
\end{proof}

\subsection{Proof of Proposition \ref{prop:psi_c=0}}
\begin{proof}
Under Assumption \ref{Assu:1},  we have that for all $x \leq y$,
\begin{align*}
    \pi^c_1(x,y)=\Prob(X\leq x,T\geq y,\Delta=1)
    &=\E\left(\ind_{\{X\leq x,Y\geq y\}}\E\left(\ind_{\{C\geq Y\}}|X,Y\right)\right)\\
    &=\E\left(\ind_{\{X\leq x,Y\geq y\}}S_{C|X}(Y)\right)\\
    &=\int_0^x\int_y^\infty S_{C|X=x'}(y')f_{XY}(x',y')dx',dy',
\end{align*}
and
\begin{align*}
    \pi^c(x,y)=\Prob(X\leq x,T\geq y)&=\E\left(\ind_{\{X\leq x\}}S_{C|X}(y)S_{Y|X}(y)\right)\\
    &=\int_0^xS_{C|X=x'}(y)S_{Y|X=x'}(y)f_{X}(x')dx'.
\end{align*}
The null hypothesis states $f_{XY}(x,y)=\tilde f_X(x)\tilde f_Y(y)$ for all $x\leq y$ such that $S_{T|X=x}(y)>0$. Thus
\begin{align*}
\pi^{c}_1(x,y)&=\int_0^x\int_y^\infty S_{C|X=x'}(y')\tilde f_{X}(x')\tilde f_{Y}(y')dx'dy',\\
\pi^{c}(x,y)&=\tilde S_{Y}(y)\int_0^xS_{C|X=x'}(y)\tilde f_{X}(x')dx'.  
\end{align*}

By using the previous result, it is easy to see that, under the null,
\begin{align*}
    &-\pi^c(x,y)\frac{\partial^2}{\partial x\partial y}\pi^c_1(x,y)=\left(\tilde S_{Y}(y)\int_0^xS_{C|X=x'}(y)\tilde f_{X}(x')dx'\right)S_{C|X=x}(y)\tilde f_X(x)\tilde f_Y(y),
\end{align*}
and
\begin{align*}
\frac{\partial\pi^c(x,y)}{\partial x} \frac{\partial\pi_1^c(x,y)}{\partial y}    
&\quad=-\left(\tilde S_{Y}(y)S_{C|X=x}(y)\tilde f_{X}(x)\right)\int_0^x S_{C|X=x'}(y)\tilde f_{X}(x')dx'\tilde f_{Y}(y)\\
&\quad=-\left(\tilde S_{Y}(y)\int_0^x S_{C|X=x'}(y)\tilde f_{X}(x')dx'\right)S_{C|X=x}(y)\tilde f_{X}(x)\tilde f_{Y}(y),
\end{align*}
from which it follows that $\rho^c=0$, and thus $\Psi=0$.
\end{proof}

\section{Proof of Theorem \ref{thm:null}}\label{sec:proof null}
Before proving Theorem \ref{thm:null} we give some essential definitions which will be used by our proofs. We will first introduce Lemma \ref{lemma:approxFinal}, which is an essential step in the proof of Theorem \ref{thm:null}. A full proof for  Lemma \ref{lemma:approxFinal} is given later in this section.

Our data are considered to live in a common filtrated probability space  $(\Omega,\mathcal{F},(\mathcal{F}_t)_{t\geq0},\Prob)$, where $\mathcal{F}$ is the natural $\sigma$-algebra, and $\mathcal{F}_t$ is the filtration generated by 
\begin{align*}
\left\{\ind_{\{T_i\leq s,\Delta_i=1\}},\ind_{\{T_i\leq s,\Delta_i=0\}},X_i:0\leq s\leq t,i\in[n]\right\},
\end{align*}
and the $\Prob$-null sets of $\mathcal{F}$. 

We define $\tau_n=\max\{T_1,\ldots,T_n\}$. For each $i\in[n]$, we define the $i$-th individual counting and risk processes, $N_i(t)$ and $Y_i(t)$, by $N_i(t)=\Delta_i\ind_{\{T_i\leq t\}}$ and $Y_i(t)=\ind_{\{T_i\geq t\}}$, respectively. For each individual $i$, we define the process $(M_i(t))_{t\geq 0}$ by 
\begin{align*}
M_i(t)=N_i(t)-\int_{(0,t]}\ind_{\{X_i\leq s\}}Y_i(s)\tilde\lambda_Y(s)ds.
\end{align*}
It is standard to verify that $M_i(t)$ is an $(\mathcal F_t)$-martingale under the null hypothesis, and that, for any bounded predictable process $(H_i(t))_{t\geq0}$,  $\int_{(0,t]}H_i(s)dM_i(s)$ is also an $(\mathcal{F}_t)$-martingale under the null hypothesis.

Let $(T_1',\Delta_1',X_1')$ and $(T_2',\Delta_2',X_2')$ be independent copies of our data $((T_i,\Delta_i,X_i))_{i=1}^n$. Sometimes our results are written in terms of $\tilde{\E}$ which is defined by $\tilde \E(\cdot)=\E\left(\cdot|((T_i,\Delta_i,X_i))_{i=1}^n\right)$. Additionally, we denote by $Y_1'$ and $Y_2'$, the individual risk functions associated to $T_1'$ and $T_2'$, which are defined by $Y_1'(t)=\ind_{\{T_1'\geq t\}}$ and $Y_2'(t)=\ind_{\{T_2'\geq t\}}$, respectively. Finally, we define $Z_i(t)=\omega(X_i,t)\ind_{\{X_i\leq t\}}$ for all $i\in[n]$, and, based on $(T_1',\Delta_1',X_1')$ and $(T_2',\Delta_2',X_2')$, we define $Z_1'(t)=\omega(X_1',t)\ind_{\{X_1'\leq t\}}$ and $Z_2'(t)=\omega(X_2',t)\ind_{\{X_2'\leq t\}}$.

\begin{lemma}\label{lemma:approxFinal}
Assume that $\mathfrak K$ is bounded. Then, under the null hypothesis
\begin{align*}
    \sqrt{n} \Psi_{n,c}=\sup_{\omega\in B_1(\mathcal{H})}\frac{1}{\sqrt{n}}\sum_{i=1}^n\int_0^{\tau_n}\left(Z_i(t)\pi^c(X_i,t)-\tilde\E\left(Z_1'(t)Y_1'(t)\ind_{\{X_i\leq X_1'\}}\right)\right)dM_i(t)+o_p(1).
\end{align*}
\end{lemma}

\subsection{Proof of Theorem \ref{thm:null}}
By Lemma \ref{lemma:approxFinal},
\begin{align*}
    \sqrt{n} \Psi_{n,c}=\sup_{\omega\in B_1(\mathcal{H})}\frac{1}{\sqrt{n}}\sum_{i=1}^n\int_0^{\tau_n}\left(Z_i(t)\pi^c(X_i,t)-\tilde\E\left(Z_1'(t)Y_1'(t)\ind_{\{X_i\leq X_1'\}}\right)\right)dM_i(t)+o_p(1).
\end{align*}

Observe that, by the reproducing kernel property, we have $Z_i(t)=\langle\omega,\mathfrak{K}((X_i,t),\cdot)\rangle_{\mathcal H}\ind_{\{X_i\leq t\}}$ and $Z_1'(t)=\langle\omega,\mathfrak{K}((X_1',t),\cdot)\rangle_{\mathcal H}\ind_{\{X_1'\leq t\}}$. Thus,
\begin{align*}
&\left(Z_i(t)\pi^c(X_i,t)-\tilde\E\left(Z_1'(t)Y_1'(t)\ind_{\{X_i\leq X_1'\}}\right)\right)\\
&=\left(\langle\omega,\mathfrak{K}((X_i,t),\cdot)\rangle_{\mathcal H}\ind_{\{X_i\leq t\}}\pi^c(X_i,t)-\tilde\E\left(\langle\omega,\mathfrak{K}((X_1',t),\cdot)\rangle_{\mathcal H}\ind_{\{X_1'\leq t\}}Y_1'(t)\ind_{\{X_i\leq X_1'\}}\right)\right)\\
&=\left\langle\omega,\mathfrak{K}((X_i,t),\cdot)\ind_{\{X_i\leq t\}}\pi^c(X_i,t)-\tilde\E\left(\mathfrak{K}((X_1',t),\cdot)\ind_{\{X_1'\leq t\}}Y_1'(t)\ind_{\{X_i\leq X_1'\}}\right)\right\rangle_{\mathcal H},
\end{align*}
where the second equality follows from the linearity of expectation, assuming Bochner integrability of the feature map (true for bounded $\mathfrak{K}$). To ease notation, we define the functions $a:\R^2\to\R$ and $b:\R^3\to\R$ by $a(X_i, t)=\ind_{\{X_i\leq t\}}\pi^c(X_i,t)$ and $b(X_1',X_i,t)=Y_1'(t)\ind_{\{X_i\leq X_1'\leq t\}}$, respectively, and write
\begin{align}
&\left(Z_i(t)\pi^c(X_i,t)-\tilde\E\left(Z_1'(t)Y_1'(t)\ind_{\{X_i\leq X_1'\}}\right)\right)\nonumber\\
&=\left\langle\omega,\mathfrak{K}((X_i,t),\cdot)a(X_i,t)-\tilde\E\left(\mathfrak{K}((X_1',t),\cdot)b(X_1',X_i,t)\right)\right\rangle_{\mathcal H}\label{eqn:diffZexpression}.
\end{align}

From the previous result, it is easy to see that
\begin{align*}
    &\frac{1}{\sqrt{n}}\sum_{i=1}^n\int_0^{\tau_n}\left(Z_i(t)\pi^c(X_i,t)-\tilde\E\left(Z_1'(t)Y_1'(t)\ind_{\{X_i\leq X_1'\}}\right)\right)dM_i(t)\\
    &=\left\langle\omega,\frac{1}{\sqrt{n}}\sum_{i=1}^n\int_0^{\tau_n}\left(\mathfrak{K}((X_i,t),\cdot)a(X_i,t)-\tilde\E\left(\mathfrak{K}((X_1',t),\cdot)b(X_1',X_i,t)\right)\right)dM_i(t)\right\rangle_{\mathcal H},
\end{align*}
and thus
\begin{align}
&\sup_{\omega\in B_1(\mathcal{H})}\left(\frac{1}{\sqrt{n}}\sum_{i=1}^n\int_0^{\tau_n}\left(Z_i(t)\pi^c(X_i,t)-\tilde\E\left(Z_1'(t)Y_1'(t)\ind_{\{X_i\leq X_1'\}}\right)\right)dM_i(t)\right)^2\nonumber\\
&\left\|\frac{1}{\sqrt{n}}\sum_{i=1}^n\int_0^{\tau_n}\left(\mathfrak{K}((X_i,t),\cdot)a(X_i,t)-\tilde\E\left(\mathfrak{K}((X_1',t),\cdot)b(X_1',X_i,t)\right)\right)dM_i(t)\right\|_{\mathcal H}^2\nonumber\\
&=\frac{1}{n}\sum_{i=1}^n\sum_{j=1}^n J((T_i,\Delta_i,X_i),(T_j,\Delta_j,X_j))\label{eqn:vstat},
\end{align}
where the function $J:(\R\times\{0,1\}\times\R)^2\to\R$ is defined by 
\begin{align*}
J((s,r,x),(s',r',x'))=\int_0^s\int_0^{s'}A((t,x),(t',x'))dm_{s',r',x'}(t')dm_{s,r,x}(t),
\end{align*}
$dm_{s,r,x}(t)=r\delta_{s}(t)-\ind_{\{s\geq t\}}\ind_{\{x\leq t\}}\tilde \lambda_Y(t)dt$ (notice that $dM_i(t)=dm_{T_i,\Delta_i,X_i}(t)$), and $A:(\R\times\R)^2\to\R$ is defined as
\begin{align*}
&A((t,x),(t',x'))\\
&=\left\langle\mathfrak{K}((x,t),\cdot)a(x,t)-\tilde\E\left(\mathfrak{K}((X_1',t),\cdot)b(X_1',x,t)\right)\right.\\
&\left.\hspace{5.2cm},\mathfrak{K}((x',t'),\cdot)a(x',t')-\tilde\E\left(\mathfrak{K}((X_2',t'),\cdot)b(X_2',x',t')\right)\right\rangle_{\mathcal{H}}\\
&=\mathfrak{K}((x,t),(x',t'))a(x,t)a(x',t')-\tilde \E(\mathfrak{K}((X_1',t),(x',t'))b(X_1',x,t)a(x',t'))\\
&-\tilde \E(\mathfrak{K}((x,t),(X_2',t'))a(x,t)b(X_2',x',t'))+ \tilde \E(\mathfrak{K}((X_1',t),(X_2',t'))b(X_1',x,t)b(X_2',x',t')).
\end{align*}

It is not difficult to verify that the sum in Equation \eqref{eqn:vstat} is a degenerate $V$-statistic. Indeed, the degeneracy property can be verified by noticing that
\begin{align*}
    &\E(J((T_i,\Delta_i,X_i),(s',r',x')))\\
    &=\E\left(\int_0^{T_i}\left(\int_0^{s'}A((t,X_i),(t',x'))dm_{s',r',x}(t')\right)dM_i(t)\right)\\
    &=\E(Q(T_i)),
\end{align*}
where $Q(s)=\int_0^{s}\left(\int_0^{s'}A((t,X_i),(t',x'))dm_{s',r',x}(t')\right)dM_i(t)$ is a zero mean $(\mathcal{F}_s)$-martingale, and thus, by the optional stopping Theorem, $\E(Q(T_i))=\E(Q(0))=0$. Then, by \cite[Theorem 4.3.2]{Koroljuk94}, we deduce 
\begin{align*}
\frac{1}{n}\sum_{i=1}^n\sum_{j=1}^n J((T_i,\Delta_i,X_i),(T_j,\Delta_j,X_j))\overset{\mathcal D}{\to}\E(J((T_1,\Delta_1,X_1),(T_1,\Delta_1,X_1)))+\mathcal Y,
\end{align*}
where $\mathcal Y=\sum_{i=1}^\infty\lambda_i(\xi^2_i-1)$, $\xi_1,\xi_2,\ldots$ are independent standard normal random variables, and $\lambda_1,\lambda_2,\ldots$ are positive constants.

The previous result, together with Lemma \ref{lemma:approxFinal}, allow us to deduce
\begin{align*}
    \Psi_{c,n}^2\overset{\mathcal D}{\to}\mu+\mathcal Y,
\end{align*}
where $\mu=\E(J((T_1,\Delta_1,X_1),(T_1,\Delta_1,X_1)))$. Notice that all  integrability conditions are satisfied as we assume the reproducing kernel is bounded. 

\subsection{Proof of Lemma \ref{lemma:approxFinal}}
In order to prove Lemma of \ref{lemma:approxFinal}, we require some intermediate results.

Recall that our test-statistic is computed as the supremum over $\omega\in B_{1}(\mathcal{H})$ of sums 
\begin{align*}
\frac{1}{n}\sum_{i=1}^n\Delta_i\omega(X_i, T_i)\widehat\pi^c(X_i,T_i)-\frac{1}{n^2}\sum_{i=1}^n\sum_{k=1}^n\Delta_k\omega(X_i,T_k)\ind_{\{X_k\leq X_i<T_k\leq T_i\}}.
\end{align*}
By using the notation introduced at the beginning of Section \ref{sec:proof null}, the previous sum can be rewritten as 
\begin{align*}
&\frac{1}{n}\sum_{i=1}^n\left(\Delta_i \omega(X_i,T_i)\widehat{\pi}^c(X_i,T_i)-\frac{1}{n}\sum_{k=1}^n\Delta_i \omega(X_k,T_i)\ind_{\{X_i\leq X_k<T_i\leq T_k\}}\right)\\
&=\frac{1}{n}\sum_{i=1}^n\int_0^{T_i}\left( \omega(X_i,y)\ind_{\{X_i\leq y\}}\widehat{\pi}^c(X_i,y)-\frac{1}{n}\sum_{k=1}^n\omega(X_k,y)\ind_{\{X_k\leq y\}}\ind_{\{y\leq T_k\}}\ind_{\{X_i\leq X_k\}}\right)dN_i(y)\\
&=\frac{1}{n}\sum_{i=1}^n\int_0^{T_i}\left( Z_i(y)\widehat{\pi}^c(X_i,y)-\frac{1}{n}\sum_{k=1}^n Z_k(y)Y_k(y)\ind_{\{X_i\leq X_k\}}\right)dN_i(y)\\
&=\frac{1}{n}\sum_{i=1}^n\int_0^{T_i}H_i(y)dN_i(y),
\end{align*}
where $H_i(y)=Z_i(y)\widehat{\pi}^c(X_i,y)-\frac{1}{n}\sum_{k=1}^n Z_k(y)Y_k(y)\ind_{\{X_i\leq X_k\}}$. Thus, 
\begin{align}
    \Psi_{n,c}=\sup_{\omega\in B_1(\mathcal{H})}\frac{1}{n}\sum_{i=1}^n\int_0^{\tau_n}H_i(y)dN_i(y)\label{eqn:PsiMartingale},
\end{align}
where recall that $\tau_n=\max\{T_1,\ldots,T_n\}$. 

\begin{proposition}\label{prop:martingale representation}
Assume that $\mathfrak K$ is bounded. Then, under the null hypothesis, the process  $(W(t))_{t\geq 0}$, defined by $W(t)=\frac{1}{n}\sum_{i=1}^n\int_0^{t}H_i(y)dN_i(y)$, is an $(\mathcal{F}_t)$-martingale, and can be rewritten as 
\begin{align*}
    W(t)&=\frac{1}{n}\sum_{i=1}^n \int_0^t H_i(y)dM_i(y).
\end{align*}
\end{proposition}

Notice that the previous proposition, and Equation \eqref{eqn:PsiMartingale} suggest the result of Lemma \ref{lemma:approxFinal}. It remains to prove that  the process $H_i(y)$ may be approximated by its ``population limit". We prove this result in two steps in the two lemmas below.

\begin{lemma}\label{lemma:approx}
Assume that $\mathfrak K$ is bounded. Then, under the null hypothesis
\begin{align*}
    \sup_{\omega\in B_1(\mathcal{H})}\frac{1}{\sqrt{n}}\sum_{i=1}^n\int_0^{\tau_n}Z_i(y)\left(\widehat\pi^c(X_i,y)-\pi^c(X_i,y)\right)dM_i(y)=o_p(1),
\end{align*}
\end{lemma}

\begin{lemma}\label{lemma:approx2}
Assume that $\mathfrak K$ is bounded. Then, under the null hypothesis
\begin{align*}
\sup_{\omega\in B_1(\mathcal{H})}\frac{1}{\sqrt{n}}\sum_{i=1}^n\int_0^{\tau_n}\left(\frac{1}{n}\sum_{j=1}^nZ_j(y)Y_j(y)\ind_{\{X_i\leq X_j\}}-\tilde\E(Z_1'(y)Y_1'(y)\ind_{\{X_i\leq X_1'\}})\right)dM_i(y)=o_p(1),
\end{align*}
\end{lemma}

\paragraph{Proof of Lemma \ref{lemma:approxFinal}:}
Equation \eqref{eqn:PsiMartingale} and  Lemma \ref{prop:martingale representation} yield
\begin{align*}
 \sqrt{n}\Psi_{n,c}=\sup_{\omega\in B_1(\mathcal{H})}\frac{1}{\sqrt{n}}\sum_{i=1}^n \int_0^{\tau_n} H_i(y) dM_i(y),    
\end{align*}
where (recall) $H_i(y)=Z_i(y)\widehat{\pi}^c(X_i,y)-\frac{1}{n}\sum_{k=1}^n Z_k(y)Y_k(y)\ind_{\{X_i\leq X_k\}}$. Notice that to obtain the result, we need to replace $\widehat \pi^c$ by its population version $\pi^c$, and, given $(T_i,\Delta_i,X_i)$, we need to replace the i.i.d. sum $\frac{1}{n}\sum_{k=1}^n Z_k(y)Y_k(y)\ind_{\{X_i\leq X_k\}}$ by its limit, which is given by $\tilde E(Z_1'(y)Y_1'(y)\ind_{\{X_i\leq X_1'\}})$. By the triangular inequality, this result follows from lemmas \ref{lemma:approx} and \ref{lemma:approx2}.

\subsection{Proofs of Proposition \ref{prop:martingale representation}, and Lemmas \ref{lemma:approx} and \ref{lemma:approx2}}

\subsubsection{Proof of Proposition \ref{prop:martingale representation}}
Recall that $dM_i(y)=dN_i(y)-\ind_{\{X_i\leq y\}}Y_i(y)\tilde \lambda_Y(y)dy$. A  straightforward computation verifies $\frac{1}{n}\sum_{i=1}^n\int_0^t H_i(y)\ind_{\{X_i\leq y\}}Y_i(y)\tilde \lambda_Y(y)dy=0$ for all $t\geq0$, and thus
\begin{align*}
W(t)
&=\frac{1}{n}\sum_{i=1}^n\int_0^t H_i(y)dM_i(y).
\end{align*}
Also, notice that $(H_i(t))_{t\geq 0}$ (with $\omega \in B_1(\mathcal H)$) is bounded and $(\mathcal{F}_t)$-predictable, and that $M_i(t)$ is an $(\mathcal F_t)-$ martingale under the null hypothesis.  Then, by standard martingale results we deduce that $(W(t))_{t\geq 0}$ is an $(\mathcal{F}_t)$-martingale.

\subsubsection{Proof of Lemma \ref{lemma:approx}}

Observe that
\begin{align*}
Z_i(t)(\widehat{\pi}^c(X_i,t)-\pi^c(X_i,t))
&=\left\langle\omega,\mathfrak K((X_i,t),\cdot)\right\rangle_{\mathcal{H}}\ind_{\{X_i\leq t\}}(\widehat{\pi}^c(X_i,t)-{\pi}^c(X_i,t))
\end{align*}
since $Z_i(t,\omega)=\omega(X_i,t)\ind_{\{X_i\leq t\}}=\langle\omega,\mathfrak K((X_i,t),\cdot)\rangle_{\mathcal{H}}\ind_{\{X_i\leq t\}}$ due to the reproducing property. 

Then, 
\begin{align*}
&\sup_{\omega\in B_1(\mathcal{H})}\left(\frac{1}{\sqrt{n}}\sum_{i=1}^n\int_0^{\tau_n}Z_i(t)\left(\widehat{\pi}^c(X_i,t)-\pi^c(X_i,t)\right) dM_i(t)\right)^2\\
&=\sup_{\omega\in B_1(\mathcal{H})}\left(\frac{1}{\sqrt{n}}\sum_{i=1}^n\int_0^{\tau_n}\left\langle\omega,\mathfrak K((X_i,t),\cdot)\right\rangle_{\mathcal{H}}\ind_{\{X_i\leq t\}}(\widehat{\pi}^c(X_i,t)-{\pi}^c(X_i,t))
 dM_i(t)\right)^2\\
&\quad=\sup_{\omega\in B_1(\mathcal{H})}\left\langle\omega,\frac{1}{\sqrt{n}}\sum_{i=1}^n\int_0^{\tau_n}\mathfrak K((X_i,t),\cdot)\ind_{\{X_i\leq t\}}(\widehat{\pi}^c(X_i,t)-{\pi}^c(X_i,t))dM_i(t)\right\rangle_{\mathcal{H}}^2\\
&\quad=\frac{1}{n}\sum_{i=1}^n\sum_{k=1}^n\int_0^{\tau_n}\int_0^{\tau_n}J((X_i,t),(X_k,s))dM_i(t)dM_k(s),
\end{align*}
where 
\begin{align}
&J((X_i,t),(X_k,s))\nonumber\\
&=\mathfrak{K}((X_i,t),(X_k,s))\ind_{\{X_i\leq t\}}\ind_{\{X_k\leq s\}}(\widehat{\pi}^c(X_i,t)-{\pi}^c(X_i,t))(\widehat{\pi}^c(X_k,s)-{\pi}^c(X_k,s))\label{eqn:Jfun}
\end{align}

Define the process $(Q(y))_{y\geq 0}$ by
\begin{align*}
    Q(y)&=\frac{1}{n}\sum_{i=1}^n\sum_{k=1}^n\int_0^{y}\int_0^{y}J((X_i,t),(X_k,s))dM_i(t)dM_k(s),
\end{align*}
and notice that we wish to prove that $Q(\tau_n)=o_p(1)$. Let $\delta>0$, then, by Markov's inequality, 
\begin{align*}
\Prob(Q(\tau_n)>\delta)&\leq \frac{\E(Q(\tau_n))}{\delta}=\frac{\E(Q_D(\tau_n))}{\delta}+\frac{2\E(Q_{D^c}(\tau_n))}{\delta},
\end{align*}
where the last equality holds since, by symmetry, $Q(y)=Q_D(y)+2Q_{D^c}(y)$, where 
\begin{align}
Q_D(y)&=\frac{1}{n}\sum_{i=1}^n\sum_{k=1}^n\int_0^{y}\int_0^{y}\ind_{\{s=t\}}J((X_i,t),(X_k,s))dM_i(t)dM_k(s),\label{eqn:diag}
\end{align}
and
\begin{align*}
Q_{D^c}(y)
&=\frac{1}{n}\sum_{k=1}^n\sum_{i=1}^n\int_0^{y}\int_{(0,s)}J((X_i,t),(X_k,s))dM_i(t)dM_k(s).
\end{align*}
By \cite[Theorem 6.8]{FerRiv2020},  $Q_{D^c}(y)$ is an $(\mathcal{F}_y)$-martingale,  and, by the optional stopping theorem, $\E(Q_{D^c}(\tau_n))=\E(Q_{D^c}(0))=0$. Thus
\begin{align*}
\Prob(Q(\tau_n)>\delta)&\leq\frac{\E(Q_D(\tau_n))}{\delta},
\end{align*}
where
\begin{align*}
Q_D(\tau_n)&=\frac{1}{n}\sum_{i=1}^n\sum_{k=1}^n\int_0^{\tau_n}J((X_i,t),(X_k,t))d[M_i,M_k](t)\\
&=\frac{1}{n}\sum_{i=1}^n\int_0^{\tau_n}J((X_i,t),(X_i,t))d[M_i](t)\\
&=\frac{1}{n}\sum_{i=1}^n\int_0^{\tau_n}J((X_i,t),(X_i,t))N_i(t)\\
&=\frac{1}{n}\sum_{i=1}^n\Delta_i J((X_i,T_i),(X_i,T_i))
\end{align*}
follows from considering continuous survival and censoring times. 

We finish the proof by proving $\E(Q_D(\tau_n))\to0$ as $n$ tends to infinity. Observe that
\begin{align*}
&\E(Q_D(\tau_n))\\
&=\E\left(\frac{1}{n}\sum_{i=1}^n\Delta_i J((X_i,T_i),(X_i,T_i))\right)\leq c_1\E\left(\frac{1}{n}\sum_{i=1}^n\Delta_i \ind_{\{X_i\leq T_i\}}(\widehat{\pi}^c(X_i,T_i)-{\pi}^c(X_i,T_i))^2\right)
\end{align*}
follows from substituting the function $J$ with the expression given in Equation \eqref{eqn:Jfun}, and by assuming the reproducing kernel is bounded by some constant $c_1>0$. By Proposition \ref{prop:sumAsmall}, the sum $\frac{1}{n}\sum_{i=1}^n\Delta_i \ind_{\{X_i\leq T_i\}}(\widehat{\pi}^c(X_i,T_i)-{\pi}^c(X_i,T_i))^2$ converges to 0 almost surely, and it is bounded by some constant $c>0$, then the desired result follows from an application of dominated convergence.

\subsubsection{Proof of Lemma \ref{lemma:approx2}}

Notice that, by the reproducing  property,
\begin{align*}
    &\frac{1}{n}\sum_{j=1}^nZ_j(t)Y_j(t)\ind_{\{X_i\leq X_j\}}-\tilde\E(Z_1'(t)Y_1'(t)\ind_{\{X_i\leq X_1'\}})\\
    &= \frac{1}{n}\sum_{j=1}^n\langle\omega, \mathfrak K ((X_j,t),\cdot)\rangle_{\mathcal H}Y_j(t)\ind_{\{X_i\leq X_j\leq t\}}-\tilde\E\left(\langle\omega,\mathfrak K((X_1',t),\cdot)\rangle_{\mathcal{H}}Y_1'(t)\ind_{\{X_i\leq X_1'\leq t\}}\right)\\
    &=\left\langle\omega,\frac{1}{n}\sum_{j=1}^n\mathfrak K ((X_j,t),\cdot)Y_j(t)\ind_{\{X_i\leq X_j\leq t\}}-\tilde\E\left(\mathfrak K((X_1',t),\cdot)Y_1'(t)\ind_{\{X_i\leq X_1'\leq t\}}\right)\right\rangle_{\mathcal{H}}.
\end{align*}
To ease notation, we define $a_{ij}(t)=Y_j(t)\ind_{\{X_i\leq X_j\leq t\}}$ and $b_{i1}'(t)=Y_1'(t)\ind_{\{X_i\leq X_1'\leq t\}}$ (similarly, we define $b_{i2}'(t)=Y_2'(t)\ind_{\{X_i\leq X_2'\leq t\}}$, where recall that $(T_1',\Delta_1',X_1')$ and $(T_2',\Delta_2',X_2')$ are independent copies of our data). Then, the previous term can be rewritten as

\begin{align*}
&\frac{1}{n}\sum_{j=1}^nZ_j(t)Y_j(t)\ind_{\{X_i\leq X_j\}}-\tilde\E(Z_1'(t)Y_1'(t)\ind_{\{X_i\leq X_1'\}})\\    
&=\left\langle\omega,\frac{1}{n}\sum_{j=1}^n\mathfrak K ((X_j,t),\cdot)a_{ij}(t)-\tilde\E\left(\mathfrak K((X_1',t),\cdot)b_{i1}'(t)\right)\right\rangle_{\mathcal{H}}.
\end{align*}

By using the fact we take supremum on the unit ball of an RKHS, it is not difficult to deduce, 
\begin{align}
&\sup_{\omega\in B_1(\mathcal{H})}\left(\frac{1}{\sqrt{n}}\sum_{i=1}^n\int_{0}^{\tau_n}\left(\frac{1}{n}\sum_{j=1}^nZ_j(t)Y_j(t)\ind_{\{X_i\leq X_j\}}-\tilde\E(Z_1'(t)Y_1'(t)\ind_{\{X_i\leq X_1'\}})\right) dM_i(y)\right)^2\nonumber\\
&=\frac{1}{n}\sum_{i=1}^n\sum_{k=1}^n\int_0^{\tau_n}\int_{0}^{\tau_n}J((X_i,t),(X_k,s))dM_i(t)dM_k(s)\label{eqn:resultapprox2},
\end{align}
where
\begin{align}
&J((X_i,t),(X_k,s))\nonumber\\
&=\frac{1}{n^2}\sum_{j=1}^n\sum_{l=1}^n\mathfrak{K}\left((X_j,t),(X_l,s)\right)a_{ij}(t)a_{kl}(s)-\frac{1}{n}\sum_{l=1}^n\tilde\E(\mathfrak{K}((X_1',t),(X_l,s))b_{i1}'(t)a_{kl}(s)\nonumber\\
&-\frac{1}{n}\sum_{j=1}^n\tilde\E(\mathfrak{K}((X_j,t),(X_2',s))a_{ij}(t)b_{k2}'(s)+\tilde\E(\mathfrak{K}((X_1',t),(X_2',s))b_{i1}'(t)b_{k2}'(s)),\label{eqn:Jfun2}
\end{align}

Following the same steps of the proof of Lemma \ref{lemma:approx}, we can prove that Equation \eqref{eqn:resultapprox2} is $o_p(1)$ by proving that
\begin{align}\label{eqn:sum_J_Lem_H**} 
\E\left(\frac{1}{n}\sum_{i=1}^n J((X_i,T_i),(X_i,T_i))\right)\to 0.
\end{align}
For this purpose, first observe that 
\begin{align*}
&\frac{1}{n}\sum_{i=1}^n J((X_i,T_i),(X_i,T_i))\\
&=\frac{1}{n^3}\sum_{i,j,l=1}^n\mathfrak{K}\left((X_j,T_i),(X_l,T_i)\right)a_{ij}(T_i)a_{il}(T_i)-\frac{2}{n^2}\sum_{i,l=1}^n\tilde\E(\mathfrak{K}((X_1',T_i),(X_l,T_i))b_{i1}'(T_i)a_{il}(T_i)\\
&+\frac{1}{n}\sum_{i=1}^n\tilde\E(\mathfrak{K}((X_1',T_i),(X_2',T_i))b_{i1}'(T_i)b_{i2}'(T_i).
\end{align*}
Each sum on the right-hand side of the previous equation is a $V$-statistic of order 3, 2 and 1, respectively. It can easily be seen that they all converge to the same limit. Consequently, the law of large numbers for $V$-statistics implies that
\begin{align*}
    \frac{1}{n}\sum_{i=1}^n J((X_i,T_i),(X_i,T_i))\to 0    
\end{align*}
almost surely. Since the reproducing kernel is assumed to be bounded and, thus, the sum is bounded as well, we can deduce, finally, \eqref{eqn:sum_J_Lem_H**} from the dominated convergence theorem.

\section{Proof of Theorem \ref{theo:consistency}} \label{sec:proofconsistency}

The consistency proof relies on the interpretation of the test statistic $\Psi_{c,n}$ and the KQIC $\Psi_c$ as the Hilbert space distances of  embeddings of certain positive measures. These distances measure the degree of (quasi)-dependence. In this spirit, this approach is connected to the well-established Hilbert Schmidt Independence Criterion, see e.g. \cite{chwialkowski2014wild, gretton2008kernel,meynaoui2019adaptive,sejdinovic2013kernel}. 

Now, let us become more concrete and introduce the following measures $\nu_0$ and $\nu_1$ on  $R^2_+$ given by
\begin{align*}
	\nu_0(dx,dy) &= \pi^c(x,y) \pi_1^c(dx,dy) \\&=  \pi^c(x,y)S_{C|X=x}(y) f_{XY}(x,y)dxdy,\\
	\\
	\nu_1(dx,dy) &= \pi^c(dx,y)\pi_1^c(x,dy)  \\&=  \left( S_{Y|X=x}(y) S_{C|X=x}(y) f_X(x) \right)\left(\int_0^x S_{C|X=t}(y) f_{XY}(t,y)dt\right) dx dy
\end{align*}
as well as their empirical counterparts $\nu_0^n$ and $\nu_1^n$ defined as
\begin{align*}
 \nu_0^n(dx,dy) &= \frac{\widehat\pi^c(x,y)}{n}  \sum_{i=1}^n\Delta_i\delta_{X_i}(x)\delta_{T_i}(y) \\
	\nu_1^n(dx,dy) &=  \frac{\ind_{\{x\leq y\}}}{n^2} \left( \sum_{i=1}^n\delta_{X_i}(x)\ind_{\{T_i\geq y\}}\right)\left(\sum_{k=1}^n\Delta_k\delta_{T_k}(y)\ind_{\{X_k\leq x\}} \right).
\end{align*}
Moreover, set $\widehat \rho^c = \nu_0^n - \nu_1^n$, which is the empirical counterpart of the measure induced by the density $\rho^c$. Then the embeddings of the (empirical) measures into the underlying RKHS are given by
\begin{align*}
    	\phi_j(\cdot) = \iint_{x\leq y} \mathfrak K((x,y), \cdot) \nu_j(dx,dy) \quad \text{ and }\quad \phi^n_j(\cdot) = \iint_{x\leq y} \mathfrak K((x,y), \cdot) \nu^n_j(dx,dy).
\end{align*}
By straightforward calculations, we obtain
\begin{align*}
    \Psi_{c,n}^2 =  \sup_{\omega\in B_1(\mathcal{H})}\left( \iint_{x\leq y}\omega(x,y)\widehat\rho^c(dx,dy) \right)^2 = \left\| \phi_0^n - \phi_1^n \right\|^2_{\mathcal{H}}
\end{align*} 
and 
\begin{align*}
    \Psi_{c}^2 =  \sup_{\omega\in B_1(\mathcal{H})}\left( \iint_{x\leq y}\omega(x,y)\rho^c(dx,dy) \right)^2 = \left\| \phi_0 - \phi_1 \right\|^2_{\mathcal{H}}.
\end{align*}
Consequently, the first part of Theorem \ref{theo:consistency} follows from convergence of the aforementioned distances:
\begin{lemma}\label{lem:conv_embeddings}
    We have $\left\| \phi_0^n - \phi_1^n \right\|^2_{\mathcal{H}} \to \left\| \phi_0 - \phi_1 \right\|^2_{\mathcal{H}}$ in probability.
\end{lemma}
The proof of Lemma \ref{lem:conv_embeddings} is given below. For the second part of Theorem \ref{theo:consistency}, recall that by assumption the chosen kernel $\mathfrak K$ is $c_0$-universal and, thus, the embedding of finite signed Borel measures is injective, see \cite{sriperumbudur2010hilbert} for details. In particular, $\Psi_{c}^2= \left\| \phi_0 - \phi_1 \right\|^2_{\mathcal{H}}$ equals zero if and only if $\nu_0\equiv \nu_1$ , or equivalently $\rho^c(x,y)=0$ for almost all $x\leq y$. Consequently, it remains to verify the following lemma, which is proven below.
\begin{lemma}\label{lem:conv_rho=0_iff}
    $\rho^c(x,y) = 0$ for almost all $x \leq y$ if and only if the null hypothesis of quasi independence is fulfilled.
\end{lemma}

\subsection{Proof of Lemma \ref{lem:conv_embeddings}}
First, observe that
\begin{align*}
	\Psi_{c,n}^2 = \sup_{\omega\in B_1(\mathcal{H})}\left( \iint_{x\leq y}\omega(x,y)\widehat\rho^c(dx,dy) \right)^2  = \left \| \phi_0^n - \phi_1^n \right\|^2_{\mathcal{H}} = V_{0,0} - 2V_{0,1} + V_{1,1},
\end{align*}
where
\begin{align*}
	&V_{0,0} = \left\| \phi_0^n \right\|^2_{\mathcal{H}} = \frac{1}{n^2} \sum_{j,i=1}^n  \mathfrak{K}((X_i,T_i),(X_j,T_j))\Delta_i\Delta_j \widehat\pi^c(X_i,T_i)\widehat\pi^c(X_j,T_j),\\
	&V_{0,1} = \langle \phi_0^n, \phi_1^n \rangle_{\mathcal H} = \frac{1}{n^3} \sum_{i,j,k = 1}^n \mathfrak{K}((X_i,T_i),(X_j,T_k)) \Delta_i \Delta_k \widehat\pi^c(X_i,T_i) \ind_{\{X_k\leq X_j < T_k \leq T_j\}}, \\
	&V_{1,1} = \left\| \phi_1^n \right\|^2_{\mathcal{H}} = \frac{1}{n^4} \sum_{i,j,k,\ell = 1}^n \mathfrak{K}((X_i,T_\ell),(X_j,T_k))  \Delta_k \Delta_\ell\ind_{\{X_\ell\leq X_i < T_\ell \leq T_i\}} \ind_{\{X_k\leq X_j < T_k \leq T_j\}}.
\end{align*}
By Proposition \ref{prop:sumAsmall} we can replace $\widehat\pi_c$ by $\pi_c$ for all asymptotic considerations, a detailed explanation for $V_{0,0}$ is given below. Thus, $V_{0,0}$, $V_{0,1}$ and $V_{1,1}$ are asymptotically equivalent to $V$-statistics of order 2, 3 and 4, respectively. For the desired statement, it remains to show that (i) $V_{0,0} \to \left\| \phi_0 \right\|^2_{\mathcal{H}}$ (ii) $V_{0,1} \to \langle \phi_0, \phi_1\rangle_{\mathcal H}$ (iii) $V_1 \to  \left\| \phi_1 \right\|^2_{\mathcal{H}}$. All three convergences follow from the strong law of large numbers for $V$-statistics and Proposition \ref{prop:sumAsmall}, as explained exemplary for (i): 

Since the kernel $\mathfrak K$ and $\widehat \pi^c$ are bounded by some $c_1>0$ and $1$, respectively, we can deduce from Proposition \ref{prop:sumAsmall} and the triangular inequality that almost surely
\begin{align*}
	&\Bigl | \frac{1}{n^2} \sum_{j,i=1}^n  \mathfrak{K}((X_i,T_i),(X_j,T_j))\Delta_i\Delta_j \Bigl( \widehat\pi^c(X_i,T_i)\widehat\pi^c(X_j,T_j) - \pi^c(X_i,T_i) \pi^c(X_j,T_j) \Bigr) \Bigr| \\
	& \leq c_1 \frac{1}{n^2} \sum_{j,i=1}^n \Delta_i\Delta_j \left(\left| \widehat\pi^c(X_i,T_i) - \pi^c(X_i,T_i)\right|  + \left|\widehat\pi^c(X_j,T_j) - \pi^c(X_j,T_j) \right|\right) \\
	& \leq  \frac{2 c_1}{n^2} \sum_{j,i=1}^n \Delta_i\Delta_j \Bigl | \widehat\pi^c(X_i,T_i) - \pi^c(X_i,T_i) \Bigr | \\
	& \leq \frac{2 c_1}{n} \sum_{i=1}^n \Delta_i \Bigl | \widehat\pi^c(X_i,T_i) - \pi^c(X_i,T_i) \Bigr | \to 0.
\end{align*}
Thus, we can replace for further asymptotic investigations $\widehat \pi^c$ by $\pi^c$. Finally, by the strong law of large numbers
\begin{align*}
	V_{0,0}& \to \E\Bigl( \mathfrak{K}((X_1,T_1),(X_2,T_2)) \Delta_1\Delta_2 \pi^c(X_1,T_1)\pi^c(X_2,T_2) \Bigr) \\
	& = \iint_{x_1<t_2} \iint_{x_2<t_2} \mathfrak{K}((x_1,t_1),(x_2,t_2)) d\nu_0(x_1,t_1) d\nu_0(x_2,t_2).
\end{align*}

\subsection{Proof of Lemma \ref{lem:conv_rho=0_iff}}
The first implication was already shown in the proof of Proposition \ref{prop:psi_c=0}. Now, assume that $\rho^c=0$. Then
\begin{align}\label{eqn:rho=0}
	\pi^c(x,y) \frac{\partial^2 \pi_1^c(x,y)}{\partial x \partial y} = \frac{ \partial\pi^c(x,y)}{\partial x}\frac{ \partial\pi_1^c(x,y)}{\partial y}.
\end{align}
Define $M(x,y) = \frac{\partial \pi_1^c(x,y)}{\partial y} = \int_0^x S_{C|X=x'}(y)f_{XY}(x',y)dx'$, then Equation \eqref{eqn:rho=0} can be rewritten as
\begin{align}\label{eqn:M/dm=pi/dpi}
	\pi^c(x,y) \frac{\partial M(x,y)}{\partial x} = \frac{\partial\pi^c(x,y)}{\partial x} M(x,y).
\end{align}
Set $Q(x,y) = \ind_{\{M(x,y)\neq 0\}}\pi^c(x,y)/M(x,y)$. From \eqref{eqn:rho=0} we can conclude that $M(x,y)=0$ implies  $\pi^c(x,y)= 0$ or
\begin{align*}
	0 = \frac{\partial^2 \pi_1^c(x,y)}{\partial x \partial y} = - S_{C|X=x}(y)f_{XY}(x,y).
\end{align*}
But, the right-hand side of the equation is positive  for all observable $(x,y)$, i.e. such that $S_{C|X=x}(y),f(x,y),f(x)>0$. Note that only these pairs are relevant and, thus, we restrict to them subsequently. Thus, $\pi^c(x,y) = Q(x,y) M(x,y)$ and differentiation with respect to $x$ leads to 
\begin{align*}
	\frac{\partial\pi^c(x,y)}{\partial x} &= \frac{\partial Q(x,y)}{\partial x} M(x,y) + Q(x,y) \frac{ \partial M(x,y)}{\partial x} \\
	 &= \frac{\partial Q(x,y)}{\partial x} M(x,y) +\frac{\pi^c(x,y)}{M(x,y)}\frac{\partial M(x,y)}{\partial x} \\
	 \,\,&= \frac{\partial Q(x,y)}{\partial x} M(x,y) + \frac{ \partial \pi^c(x,y)}{\partial x}.
\end{align*}
Thus, $\partial Q(x,y)/\partial x = 0$ for all (observable) $x \leq y$. In particular, $Q$ does not depend on $x$, and we can write $Q(y)$ instead of $Q(x,y)$. Consequently, we can deduce from the definitions of $Q$, $M$ and $\pi^c$ that 
\begin{align*}
	- Q(y)\int_0^x S_{C|X=x'}(y)f_{XY}(x',y)dx' = \int_0^x S_{Y|X=x'}(y)S_{C|X=x'}(y)f_{X}(x')dx'.
\end{align*}
In particular, we can deduce that for all observable $x \leq y$
\begin{align*}
	- Q(y) S_{C|X=x}(y)f_{XY}(x,y) =  S_{Y|X=x}(y)S_{C|X=x}(y)f_{X}(x).
\end{align*}
From this we obtain
\begin{align*}
	f_{XY}(t,y) &= - Q(y)^{-1} S_{Y|X=t}(y)f_{X}(t)\\
	\Leftrightarrow 
	f_X(t) f_{Y|X=t}(y) &= - Q(y)^{-1} S_{Y|X=t}(y)f_{X}(t) \\
	\Leftrightarrow 
	\lambda_{Y|X=x}(y) &= - Q(y)^{-1},
\end{align*}
where $\lambda_{Y|X=x}$ denotes the hazard rate function, which does not depend on $x$. Note that $S_{Y|X=x}(x)=1$ and
\begin{align*}
	S_{Y|X=x}(y) = \frac{ S_{Y|X=x}(y) }{ S_{Y|X=x}(x)} = \exp\Bigl(   \int_x^y Q(s)^{-1} ds \Bigr).
\end{align*}
Moreover, for $t<x< y$
\begin{align*}
	\frac{ S_{Y|X=x}(y) }{S_{Y|X=t}(y)} &= \exp\Bigl(  \int_x^y Q(s)^{-1}ds - \int_t^y Q(s)^{-1}ds  \Bigr) = \frac{g(t)}{g(x)},
\end{align*}
where $g(x) = \exp(\int_{l_{X}}^x Q(s)^{-1}ds)$ and $l_X = \inf\{s\geq 0:f_X(s)>0\}$ is the lower bound of the support of $X$ (given $X\leq Y$). Differentiation with respect to $y$ leads to
\begin{align*}
	f_{Y|X=x}(y) = f_{Y|X=t}(y) \frac{g(t)}{g(x)}
\end{align*}
and, thus,
\begin{align}\label{eqn:f_XY_final_proof}
	f_{XY}(x,y) = \frac{f_{XY}(t,y)g(t)}{f_X(t)} \frac{f_X(x)}{g(x)}.
\end{align}
Now, let $(t_n)_{n\in\N}$ be a strictly decreasing sequence with $f(t_n)>0$ and $t_n\to l_X$ as $n\to \infty$. Set $t_0=\infty$. Then we can deduce from Equation \eqref{eqn:f_XY_final_proof} that 
\begin{align*}
	&f_{XY}(x,y) = \widetilde f_Y(y) \widetilde f_X(x),
\end{align*}
where
\begin{align*}
	 \widetilde f_Y(y) = \sum_{n=1}^\infty \frac{f_{XY}(t_n,y)g(t_n)}{f_X(t_n)} \ind_{\{y\in(t_{n}, t_{n-1})\}},\quad
	 \widetilde f_X(x)  = \frac{f_X(x)}{g(x)}.
\end{align*}

\section{Efficient implementation of wild bootstrap}\label{sec:efficient}
Similarly to the work of \cite{chwialkowski2014wild}, we can implement our Wild-Bootstrap efficiently by considering the identity $\text{trace}(\boldsymbol A \boldsymbol B) = \sum_{ij} (\boldsymbol{A} \odot \boldsymbol{B})_{ij}$, where $\boldsymbol A$ and $\boldsymbol B$ denote $n\times n$ matrices, and $\odot$ denotes the element-wise product. By using this identity our test-statistic can be written as 
\begin{align*}
 \Psi_{c,n}^2
&=\frac{1}{n^2}\text{trace}(\boldsymbol{K}\boldsymbol{\widehat{\pi}}^c\boldsymbol{\tilde L}\boldsymbol{\widehat{\pi}}^c-2{\boldsymbol K}\boldsymbol{\widehat{\pi}}^c{\boldsymbol{\tilde L}}\boldsymbol{{B}}^\intercal+\boldsymbol{K}\boldsymbol{ B}\boldsymbol{{ L}}\boldsymbol{{B}}^\intercal)\\
&=\sum_{ij}\left(\boldsymbol K\odot (\boldsymbol{\widehat{\pi}}^c\boldsymbol{\tilde L}\boldsymbol{\widehat{\pi}}^c-2\boldsymbol{\widehat{\pi}}^c{\boldsymbol{\tilde L}}\boldsymbol{{B}}^\intercal+\boldsymbol{ B}\boldsymbol{{ L}}\boldsymbol{{B}}^\intercal)\right)_{ij}\\
&=\sum_{ij}\boldsymbol{M}_{ij},
\end{align*}
where $\boldsymbol M= \boldsymbol K\odot (\boldsymbol{\widehat{\pi}}^c\boldsymbol{\tilde L}\boldsymbol{\widehat{\pi}}^c-2\boldsymbol{\widehat{\pi}}^c{\boldsymbol{\tilde L}}\boldsymbol{{B}}^\intercal+\boldsymbol{ B}\boldsymbol{{ L}}\boldsymbol{{B}}^\intercal)$ is a $V$-statistic matrix. Then, the wild-bootstrap version of the preceding $V$-statistic is $(\Psi_{c,n}^{\text{WB}})^2=W^\intercal \boldsymbol M W $ where $W = (W_1, \dots, W_n) \in \R^n$ are the wild bootstrap weights. In this way, we only need to compute $O(n^2)$ sum once, for each wild bootstrap, instead of computing several (actually 6 times) $O(n^3)$ matrix multiplications and two $O(n^2)$ matrix multiplications for $K^W$.

\section{Review of related quasi independence tests}\label{sec:method_review}
In this section, we review the quasi-independence tests implemented in Section \ref{sec:Experiments} of the main text.

\textbf{WLR} refers to the weighted log-rank test discussed in \cite{EmuraWang2010}, which is defined as
\begin{align*}
    L_W = \int_{x\leq y} W(x,y)\left \{ N_{11}(dx,dy) - \frac{N_{1 {\bullet}}(dx,y) N_{{\bullet}1}(x,dy)}{R(x,y) } \right\},
\end{align*}
where  
$$
N_{11}(dx,dy) = \sum_j \ind(X_j=x, T_j= y, \Delta_j=1),
\quad
N_{{\bullet}1}(x,dy) = \sum_j \ind(X_j \leq x, T_j = y, \Delta_j=1),
$$
$$
N_{1 {\bullet}}(dx,y) = \sum_j \ind(X_j = x, T_j\geq y),
\quad
R(x,y) = \sum_j \ind(X_j\leq x, T_j \geq y),$$
and $W:\R_+^2\to\R$ is the weight function given by $W(x,y)=R(x,y)$. We note that, $R(x,y)=n\widehat\pi^c(x,y)$ defined in our notation. It is straightforward to see $\Psi_{c,n}^2=\frac{1}{n^2}L_W^2$ in the case $\mathfrak K =1 $.

\textbf{WLR\_SC} refers to the previous log-rank test with weight $W$ given by $W(x,y)=\int_0^{x} \widehat{S}_{C_R}((y-u)-)^{-1}\widehat\pi^c(du,y)$, where $\widehat S_{C_R}$ is the Kaplan-Meier estimator based on the data $((C_i-X_i,1-\Delta_i))_{i=1}^n$. tis specific test was proposed to the general assumption $Y_i\perp C_i|X_i$.

\textbf{M\&B} refers to the conditional Kendall's tau statistic in discussed in \cite{MartinBetensky2005}. Let 
\begin{align*}
B_{ij} &= \{\max(X_i,X_j) \leq \min(T_i,T_j)\}\\
&\quad\cap \{(\Delta_i=\Delta_j=1) \cup (T_j>T_i,\Delta_i=1,\Delta_j=0)\cup(T_i>T_j,\Delta_i=1,\Delta_j=0)\}.    
\end{align*}
The conditional Kendall's tau statistic is given by 
\begin{align*}
\widehat\tau_b&=\sum_{i<j}\ind_{\{B_{ij}\}}\text{sign}((X_i-X_j)(T_i-T_j)).
\end{align*}
\textbf{MinP1} and \textbf{MinP2} refers to the minimal p-value selection tests which are permutation based methods proposed in \cite{chiou2018permutation}. These tests are based on the underlying principle that, under quasi-independence, the distributions of $Y|X\leq t$ and $Y|X>t$ should not differ, where $t$ denotes some cut-point. Given a collection of possible cut-points $t$, the authors perform several two-sample log-rank tests for comparing $\{(T_i,\Delta_i):X_i\leq t\}$ and $\{(T_i,\Delta_i):X_i> t\}$ (under right-censored data), and set as their test-statistic the minimum log-rank $p$-value obtained. To guarantee meaningful comparisons, the authors consider cut-points that yield at least $E$ events in each group.

The first test proposed is the following:

\tb{MinP1}: 
\begin{itemize}
    \item[1] Set $m=0$
    \item[2] Set $m=m+1$ and split the data into two groups $\{i:X_i\leq X_m\}$ and $\{i:X_i> X_m\}$.
    \item[3]Check the groups are admissible by verifying $E\leq \sum_{i=1}^n\Delta_i\ind_{\{X_i<X_m\}}\leq n-E$. If the latter holds, perform a two-sample log-rank test for comparing $\{(T_i,\Delta_i):X_i\leq X_m\}$ and $\{(T_i,\Delta_i):X_i> X_m\}$, and record the $p$-value obtained. If the condition is not satisfied, record a $p$-value equal to 1.
    \item[4] If $m<n$ return to Step 2
    \item[5] Set as test-statistic $minp_1$ the smallest p-value obtained. 
\end{itemize}

Alternatively, the authors propose a second test, which splits the data according to whether or not, the entry times belong to the interval $(t-\epsilon,t+\epsilon)$, where $t$, again, denotes a cut-point  and $\epsilon>0$. Similarly to the previous case, we need to ensure that each group contains at least $R$ data points, this can be done by choosing a suitable $\epsilon>0$.

\tb{MinP2}: 
\begin{itemize}
    \item[1] Set $m=0$
    \item[2] Set $m=m+1$ and split the data into two groups $\{i:X_i\in (X_m-\epsilon_m,X_m+\epsilon_m)\}$ and $\{i:X_i\not\in (X_m-\epsilon_m,X_m+\epsilon_m)\}$, where $\epsilon_m$ is the smallest $\epsilon>0$ such that there are at least $E$ data-points in each group. Record the value $\epsilon_m$.
    \item[3]If $m<n$ return to Step 2. 
    \item[4]Set $\epsilon=\max_{m}\epsilon_m$ and $m=0$
    \item[5]Set $m=m+1$. Verify $E\leq \sum_{i=1}^n\Delta_i\ind_{\{T_m-\epsilon<T_i<T_m+\epsilon\}}\leq n-E$ which checks that the partition of the data is admissible (under right-censoring). If the latter holds, perform a two-sample log-rank test  for comparing each group and record the $p$-value. If the partition is not admissible record a $p$-value equal to 1.
    \item[6] If $m<n$ return to Step 5.
    \item[7] Set as test-statistic $minp_2$ the smallest p-value obtained.
\end{itemize}

The rejection regions for these tests are computed by using a permutation approach.

\section{Additional discussions for empirical results}\label{app:experiments}
In this section, we provide additional information and discussions on our empirical findings.

\subsection{Computational runtime}

As shown in Table \ref{tab:run_time}, our proposed test, implemented as described in Appendix \ref{sec:efficient}, has a significantly lower runtime when compared with the permutation approaches which require much longer run-time. M\&B 
 implements the conditional Kendall's tau statistic which has a closed-form expression for the null distribution, therefore the runtime is much lower again.

\begin{table}[h!]
    \centering
    \begin{tabular}{lrrrrrrrrr}
\toprule
{\quad n} &     100 &     200 &     300 &      400 &      500 &      600 &      700 &      800 &      900 \\
\midrule
KQIC &   0.012 &   0.019 &   0.031 &    0.041 &    0.063 &    0.085 &    0.130 &    0.152 &    0.200 \\
MinP1 &  15.77 &  41.62 &  56.61 &   90.52 &  113.7 &  154.4 &  254.4 &  299.2 &  389.1 \\
MinP2 &  20.33 &  35.08 &  59.09 &  101.4 &  123.7 &  174.3 &  242.4 &  300.9 &  354.2 \\
M\&B   &   0.002 &   0.002 &   0.002 &    0.003 &    0.004 &    0.006 &    0.006 &    0.009 &    0.021 \\
\bottomrule
\end{tabular}
\vspace{0.3cm}
\caption{The runtime, in seconds, for a single trial using $500$ wild bootstrap samples for KQIC and $500$ permutations for MinP1 and MinP2. M\&B does not require to approximate the null distribution.}
    \label{tab:run_time}
\end{table}

\subsection{Kernel choice}
\paragraph{Parameter selection} In kernel-based hypothesis testing, test power (i.e.,  the probability of rejecting $H_0$ when it is false) can vary for different choices of kernel parameters, such as the bandwidth in Gaussian kernels \cite{gretton2012optimal}.
Previous works \cite{gretton2012optimal,jitkrittum2018informative,jitkrittum2016interpretable,jitkrittum2017linear,sutherland2016generative} have proposed  to choose the kernel parameters by maximizing a proxy for the test power. In the uncensored setting, the test power is (to a good approximation) increased by maximising the ratio of the test statistic to its standard deviation under the alternative.  We conjecture that the same ratio represents a good criterion in the setting of left-truncation and right-censoring, for which we have strong empirical evidence.  A formal proof remains a topic for future work.

In the censored case, the test power criterion takes the form $\frac{\Psi_{c}^2}{\sigma_{H_1}}$, where $\sigma_{H_1}$ is the standard deviation of $\Psi_{c}^2$ under the alternative hypothesis $H_1$.  Thus, to maximise the test power, we choose the kernel parameter $\theta$ by
$$
\theta^{\ast} = \argmax_{\theta} \frac{{\Psi_{c}^2}}{\sigma_{H_1}}.
$$
In practice, we use part of the data to  compute  ${\Psi_{c,n}^2}/{(\widehat{\sigma}_{H_1}+\lambda)}$, where $\widehat{\sigma}_{H_1}$ is an empirical estimate of ${\sigma}_{H_1}$ and a regularisation parameter $\lambda>0$ is added for numerical stability. 
We then perform the test on the remaining data with the selected $\theta^{\ast}$. 
A $20/80$ train-test split is suggested in \cite{jitkrittum2017linear}  for learning the parameter. We use the regulariser $\lambda=0.01$. 

We next give our empirical estimate  for the variance $\widehat\sigma^2_{H_1}$. First,  $\Psi_{c,n}^2$ can be written as $\Psi_{c,n}^2=\frac{1}{n^2}\sum_{i=1}^n\sum_{j=1}^n J_n((T_i,\Delta_i,X_i),(T_j,\Delta_j,X_j))$, where $J_n$ is defined by
\begin{align*}
J_n((T_i,\Delta_i,X_i),(T_j,\Delta_j,X_j))&=\Delta_i\Delta_j L(T_i,T_j)g_n(X_i,X_j),
\end{align*}
where 
\begin{align*}
    g_n(X_i,X_j)=K(X_i,X_j)\boldsymbol{\widehat\pi^c}_{ii}\boldsymbol{\widehat\pi^c}_{jj}-2\sum_{l=1}^nK(X_i,X_l) \boldsymbol{\widehat\pi^c}_{ii} \boldsymbol{B}_{l,j}+\sum_{l=1}^n\sum_{k=1}^nK(X_k,X_l)\boldsymbol{B}_{k,i}\boldsymbol{ B}_{l,j},
\end{align*}
and $\boldsymbol{\widehat\pi^c}_{ii}=\widehat\pi^c(X_i,T_i)$ and $B_{k,i}=\ind_{\{X_i\leq X_k<T_i\leq T_k\}}/n$. This ``$V$-statistic" form suggests that the variance can be estimated by

\begin{align*}
 \widehat\sigma^2_{H_1}
 &=\frac{1}{n}\sum_{i=1}^n\left(\frac{1}{n}\sum_{j=1}^n J_n(i,j)\right)^2-\left(\frac{1}{n^2}\sum_{i=1}^n\sum_{j=1}^nJ_n(i,j)\right)^2,
\end{align*}
where $J_n(i,j)=J_n((T_i,\Delta_i,X_i),(T_j,\Delta_j,X_j)).$

Finally, some remarks on the performance of our kernel selection heuristic in experiments. For simple cases, our kernel selection procedure makes little difference, since a broad range of kernel bandwidths    yields good results, and the ``median heuristic'' (selection of the bandwidth as the pairwise inter-sample distance) is  adequate. On the other hand, our procedure results in large power improvements for more complex cases such as periodic dependency at high frequencies,  where the median distance between samples does not correspond to the lengthscale at which dependence occurs. Similar phenomena have also been observed previously in  \cite{sutherland2016generative}.

\paragraph{Inverse Multi-Quadratic (IMQ) kernel} We further study the performance of the IMQ kernel on our proposed test. The IMQ kernel has the form $k(x,y) = (c^2 + \|x-y\|^2)^b$, for constant $c>0$ and $b\in(-1,0)$. As proposed in \cite{gorham2017measuring}, we choose $b = -\frac{1}{2}$. We select the parameter $c$ by maximizing a heuristic proxy for test power, as discussed above. The controlled Type-I error is shown in Table \ref{tab:imq_h0}, where $X$ and $Y$ are independent samples from $\operatorname{Exp}(1)$. Truncation and right-censoring apply with censoring time independently generated from exponential distribution. We report the test power of KQIC with IMQ kernel in  later sections.

\begin{table}[h]
    \centering
\begin{tabular}{lrrrrrrrrrr}
\toprule
{\qquad n} &   50  &   100 &   150 &   200 &   250 &   300 &   350 &   400 &   450 &   500 \\
\midrule
KQIC\_IMQ &  0.08 &  0.05 &  0.03 &  0.05 &  0.04 &  0.05 &  0.05 &  0.07 &  0.07 &  0.05 \\
\bottomrule
\\
\end{tabular}
    \caption{Type-I error for IMQ kernels, with $\alpha=0.05$, censoring level $25\%$, $100$ trials, and increasing sample size $n$.}
    \label{tab:imq_h0}
\end{table}

\subsection{Periodic dependencies}\label{app:experiments_period}

As briefly mentioned in the main text, the parameter $\beta$ controls the frequency of  sinusoidal dependence. At a given  sample size,  dependence  becomes harder to detect  as the frequency $\beta$ increases, both for  our test  and for  competing methods. We illustrate the datasets visually in Figure \ref{fig:periodic_sample_by_size}. 
For a fixed sample size, the test power decreases as  frequency increases, which is observed in our results  in  Figure \ref{fig:periodic_res}. For high frequency cases, larger sample size is required to correctly reject the null as shown in Figure \ref{fig:periodic_res_supp}.

\begin{figure}[h!]
\centering
\caption{Samples from periodic dependency model w.r.t. frequency coefficient $\beta$.} \label{fig:periodic_sample_by_size}
\includegraphics[width=1.0\textwidth]{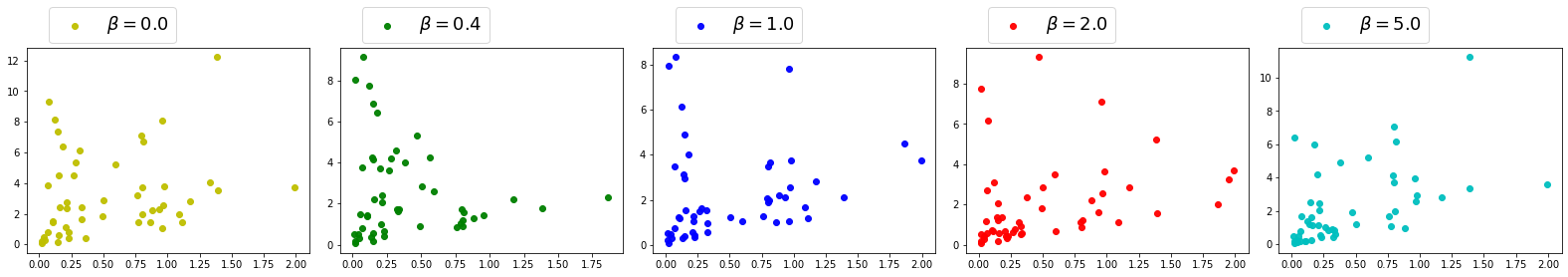}\subcaption{Sample size: n = 50}

\includegraphics[width=1.0\textwidth]{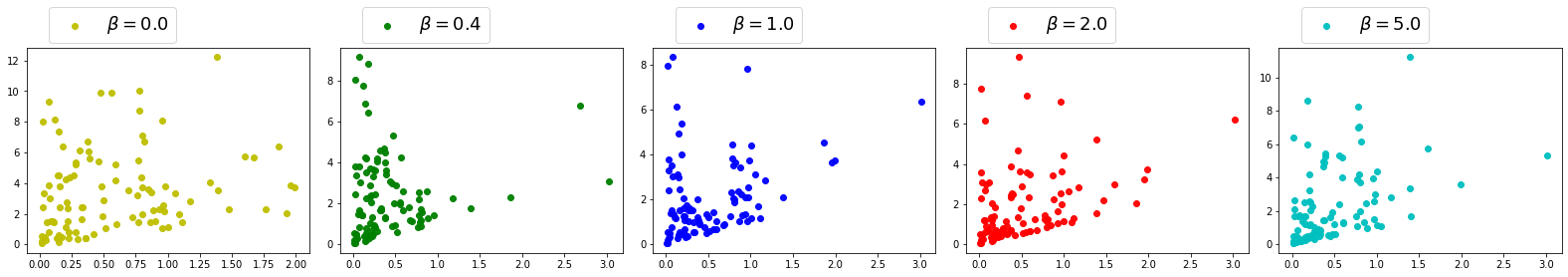}\subcaption{Sample size: n = 100}

\includegraphics[width=1.0\textwidth]{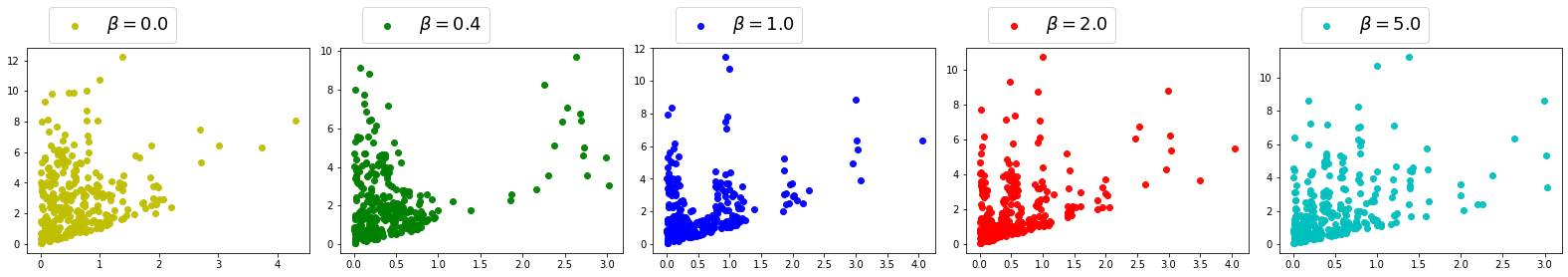}\subcaption{Sample size: n = 300}
\end{figure}

Type-I error is reported in Table \ref{tab:periodic_h0}, and is close to the desired level (subject to finite sample effects).

\begin{table}[h!]
    \centering
\begin{tabular}{lrrrrrrrrrr}
\toprule
{\qquad n } &   100  &   300  &   500  &   700  &   900  &   1100 &   1300 &  1500 &   1700 &  1900 \\
\midrule
KQIC\_Gauss &  0.045 &  0.060 &  0.055 &  0.040 &  0.045 &  0.045 &  0.040 &  0.030 &  0.045 &  0.050 \\
KQIC\_IMQ &  0.050 &  0.055 &  0.045 &  0.030 &  0.020 &  0.040 &  0.025 &  0.020 &  0.015 &  0.020 \\
WLR   &  0.030 &  0.045 &  0.050 &  0.025 &  0.045 &  0.015 &  0.015 &  0.030 &  0.025 &  0.040 \\
WLR\_SC &  0.035 &  0.060 &  0.030 &  0.025 &  0.060 &  0.070 &  0.045 &  0.055 &  0.050 &  0.060  \\
\bottomrule
\\
\end{tabular}
    \caption{Type-I error with $\alpha=0.05$, censoring level $25\%$, $200$ trials, and increasing sample size $n$.}
    \label{tab:periodic_h0}
\end{table}

\subsection{Dependent censoring}\label{app:dependent censoring}
In this section we show that our test achieves correct Type-I error under the null hypothesis even when considering dependent censoring times $C$. As stated in Assumption \ref{Assu:1}, we only require $Y\perp C|X$, which is a standard assumption, as also  considered in \cite{EmuraWang2010}.

\begin{figure}[t!]
    \centering
    \includegraphics[width=1.0\textwidth]{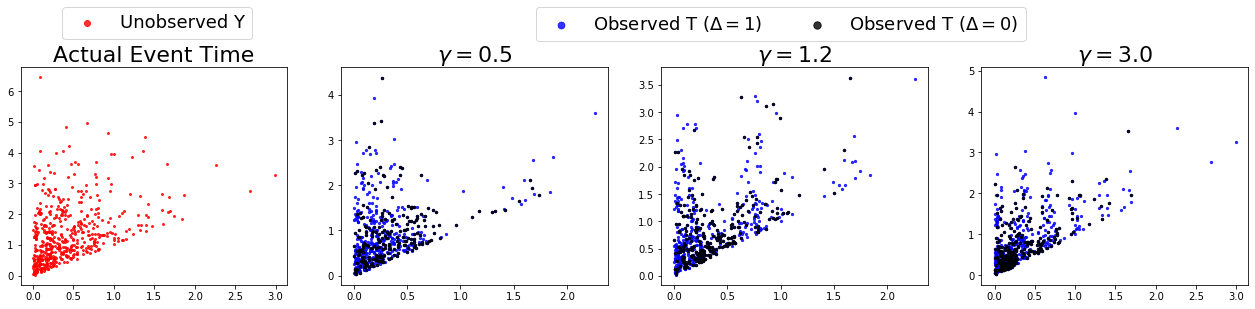}
    \caption{Samples generated from $H_0$ with periodic dependent censoring distributions.}
    \label{fig:depend_censor}
\end{figure}

We generate the data as follows: Sample $X_i \sim \textnormal{Exp}(1)$, then generate $Y_i \sim \textnormal{Exp}(1)$ (independent of $X_i$) and $C_i|X_i \sim \textnormal{Exp}(e^{\cos (2 \pi \gamma  X_i})$. Generate the observed data point $(T_i,\Delta_i, X_i)$, where $T_i=\min\{Y_i,C_i\}$ and $\Delta_i=\ind_{\{T_i=Y_i\}}$ and keep it as a valid sample only if $T_i\geq X_i$. Notice that in this case both left truncation and right-censoring are present in the data. Also, notice that the null hypothesis holds since the survival times $Y_i$ are quasi-independent of the entry times $X_i$. In Figure \ref{fig:depend_censor}, we show the unobserved pairs $(X,Y)$ and the observed pairs $(X,T)$ where the censoring variable is generated using different censoring frequencies $\gamma$. From the plot, we see that the entry times $X$ and survival times $Y$ look quasi-independent, but, due to the periodic dependency of the censoring distribution, the observed data $(X,T)$ show a periodic trend, which looks similar to the observations in Figure \ref{fig:periodic_sample_by_size}. However, since this dependency is due to the censoring times $C$ instead of the survival times $Y$, our tests are able to recover $H_0$ and achieve correct test level, as shown in Table \ref{tab:depend_censor_h0}. The tests proposed in \cite{EmuraWang2010} are also valid under Assumption \ref{Assu:1}, thus we include the results for WLR and WLR\_SC as well. From Table \ref{tab:depend_censor_h0}, we observe that KQIC with both Guassian and IMQ kernels, as well as WLR achieve the correct test level; however, WLR\_SC has slightly higher type-I errors when sample size is small and achieves correct test-level when sample size becomes large (recall that WLR\_SC uses a data dependent weight, thus convergence in this case might be slower).

\begin{table}[t!]
    \centering
\caption{ Type-I error for periodic dependent censoring distributions, with $\alpha=0.05$ and $100$ trials.
}\label{tab:depend_censor_h0}

\vspace{0.3cm}

\begin{tabular}{lrrrrrrrrrr}
\toprule
{\qquad n} &     100 &     200 &     300 &     400 &     500&     600 &     700 &     800 &     900 &     1000 \\
\midrule
KQIC\_Gauss &  0.07 &  0.06 &  0.03 &  0.03 &  0.06 &  0.05 &  0.04 &  0.04 &  0.03 &  0.07 \\
KQIC\_IMQ &  0.07 &  0.06 &  0.04 &  0.01 &  0.03 &  0.04 &  0.05 &  0.05 &  0.06 &  0.07 \\
WLR &  0.07 &  0.05 &  0.03 &  0.01 &  0.03 &  0.04 &  0.05 &  0.04 &  0.03 &  0.07 \\
WLR\_SC &  0.10 &  0.08 &  0.09 &  0.13 &  0.13 &  0.04 &  0.09 &  0.05 &  0.04 &  0.06 \\
\bottomrule
\end{tabular}
\subcaption{Censoring frequency $\gamma = 0.5$. Censoring level $30\%$}\label{tab:depend_censor_h0_5}
\begin{tabular}{lrrrrrrrrrr}
\toprule
{\qquad n} &     100 &     200 &     300 &     400 &     500&     600 &     700 &     800 &     900 &     1000 \\
\midrule
KQIC\_Gauss &  0.03 &  0.02 &  0.01 &  0.04 &  0.05 &  0.06 &  0.05 &  0.04 &  0.06 &  0.04 \\
KQIC\_IMQ &  0.02 &  0.03 &  0.03 &  0.03 &  0.04 &  0.05 &  0.05 &  0.04 &  0.04 &  0.04 \\
WLR &  0.02 &  0.02 &  0.03 &  0.05 &  0.04 &  0.04 &  0.04 &  0.04 &  0.05 &  0.05 \\
WLR\_SC &  0.06 &  0.12 &  0.17 &  0.15 &  0.10 &  0.11 &   0.06 &  0.04 &  0.05 &  0.05 \\
\bottomrule
\end{tabular}
\subcaption{Censoring frequency $\gamma = 1.2$. Censoring level $35\%$}\label{tab:depend_censor_h0_1}
\begin{tabular}{lrrrrrrrrrr}
\toprule
{\qquad n} &     100 &     200 &     300 &     400 &     500&     600 &     700 &     800 &     900 &     1000 \\
\midrule
KQIC\_Gauss &  0.06 &  0.04 &  0.06 &  0.02 &  0.02 &  0.04 &  0.03 &  0.03 &  0.05 &  0.04 \\
KQIC\_IMQ &  0.05 &  0.04 &  0.05 &  0.01 &  0.02 &  0.04 &  0.04 &  0.03 &  0.03 &  0.04 \\
WLR &  0.04 &  0.02 &  0.04 &  0.02 &  0.02 &  0.04 &  0.04 &  0.03 &  0.05 &  0.03 \\
WLR\_SC &  0.09 &  0.10 &  0.13 &  0.15 &  0.10 &  0.08 &  0.05 &  0.03 &  0.03 &  0.04 \\
\bottomrule
\end{tabular}
\subcaption{Censoring frequency $\gamma = 3.0$. Censoring level $40\%$}\label{tab:depend_censor_h0_3}
\vspace{0.3cm}
\end{table}

\subsection{Test performance w.r.t. censoring level}\label{app:censor_level}


We report the Type-I error for different censoring percentages,  see Table \ref{tab:censor_level_h0}. With reasonable censoring level (e.g. < 90\%), the Type-I errors are well controlled. WLR\_SC has higher Type-I with small sample sizes, which is similarly observed in Table \ref{tab:depend_censor_h0}.  However, the Type-I error is less controlled at extremely high censoring percentages, due to the lack for useful information obtained. In practise, we may need to be careful dealing with extremely high censoring when applying the quasi-independence tests.



\begin{table}[ht!]
\centering
\begin{tabular}{lrrrrrrr}
\toprule
{$\%$ censored} & 20 &  35 &  50 &  70 &  85 &  92 &  95 \\
\midrule
{$n=200$} \\
\midrule
KQIC\_Gauss &  0.040 &  0.025 &  0.015 &  0.045 &  0.035 &  0.085 &  0.115 \\
KQIC\_IMQ   &  0.040 &  0.060 &  0.050 &  0.055 &  0.070 &  0.100 &  0.185 \\
WLR        &  0.055 &  0.035 &  0.040 &  0.050 &  0.030 &  0.075 &  0.120 \\
WLR\_SC     &  0.045 &  0.105 &  0.075 &  0.120 &  0.060 &  0.035 &  0.075 \\
\midrule
\end{tabular}

\begin{tabular}{lrrrrrrr}
{$n=300$} \\
\midrule
KQIC\_Gauss &  0.055 &  0.040 &  0.055 &  0.045 &  0.060 &  0.090 &  0.065 \\
KQIC\_IMQ   &  0.045 &  0.050 &  0.070 &  0.050 &  0.050 &  0.105 &  0.115 \\
WLR        &  0.030 &  0.055 &  0.055 &  0.040 &  0.050 &  0.075 &  0.065 \\
WLR\_SC     &  0.080 &  0.120 &  0.140 &  0.095 &  0.125 &  0.095 &  0.025 \\
\midrule
\end{tabular}

\begin{tabular}{lrrrrrrr}
{$n=500$} \\
\midrule
KQIC\_Gauss &  0.040 &  0.050 &  0.035 &  0.030 &  0.030 &  0.050 &  0.090 \\
KQIC\_IMQ   &  0.065 &  0.030 &  0.050 &  0.040 &  0.080 &  0.100 &  0.050 \\
WLR        &  0.035 &  0.035 &  0.050 &  0.035 &  0.040 &  0.060 &  0.075 \\
WLR\_SC     &  0.060 &  0.035 &  0.055 &  0.075 &  0.065 &  0.035 &  0.015 \\
\midrule
\end{tabular}

\begin{tabular}{lrrrrrrr}
{$n=800$} \\
\midrule
KQIC\_Gauss &  0.045 &  0.030 &  0.030 &  0.065 &  0.030 &  0.065 &  0.080 \\
KQIC\_IMQ   &  0.065 &  0.050 &  0.050 &  0.060 &  0.060 &  0.090 &  0.140 \\
WLR        &  0.015 &  0.010 &  0.025 &  0.055 &  0.065 &  0.085 &  0.100 \\
WLR\_SC     &  0.095 &  0.040 &  0.065 &  0.080 &  0.075 &  0.045 &  0.025 \\
\bottomrule
\end{tabular}
\vspace{0.3cm}
    \caption{Type-I error for different censoring level, with $\alpha=0.05$ and $200$ trials,}\label{tab:censor_level_h0}
    \vspace{0.3cm}
\end{table}

\end{document}